\documentclass[11pt,a4paper]{article}

\usepackage[utf8]{inputenc} 
\usepackage{amsmath,amssymb,amsbsy,bm,latexsym,mathrsfs}
\usepackage{xspace, tikz, booktabs, multirow, graphicx,utf8math}
\usetikzlibrary{arrows}
\usepackage[english]{babel}
\usepackage{amsthm}
\usepackage[margin=1cm,font=small]{caption}
\usepackage{anysize}

\theoremstyle{plain}
\newtheorem{theorem}{Theorem}
\newtheorem{lemma}[theorem]{Lemma}
\newtheorem{corollary}[theorem]{Corollary}

\theoremstyle{definition}
\newtheorem{remark}{Remark}

\providecommand{\R}{\ensuremath{\mathbb{R}}}
\providecommand{\Z}{\ensuremath{\mathbb{Z}}}
\providecommand{\s}{\ensuremath{\mathbb{S}}}
\renewcommand{\epsilon}{\varepsilon}
\providecommand{\dv}{\ensuremath{\mathrm{div}\xspace}}
\providecommand{\cl}{\ensuremath{\mathrm{cl}\xspace}}
\providecommand{\st}{\ensuremath{\mathrm{st}\xspace}}
\providecommand{\bp}{\ensuremath{\mathrm{bp}\xspace}}
\providecommand{\opt}{\ensuremath{\mathrm{OPT}\xspace}}

\title{Diversity maximization in doubling metrics%
}
\author{
  Alfonso Cevallos\footnote{Swiss Federal Institute of Technology (ETH), Switzerland. \texttt{alfonso.cevallos@ifor.math.ethz.ch}} 
	\quad \quad  
	Friedrich Eisenbrand\footnote{\'Ecole Polytechnique Fédérale de Lausanne (EPFL), Switzerland. \texttt{friedrich.eisenbrand@epfl.ch}} 
	\quad \quad 
	Sarah Morell\footnote{Technische Universität Berlin (TU Berlin), Germany. \texttt{morell@math.tu-berlin.de}} \\[.1cm]
}
\date{\today }

\begin{document}
\maketitle

\begin{abstract} 
Diversity maximization is an important geometric optimization problem
with many applications in recommender systems, machine learning or
search engines among others. A typical diversification problem is as
follows:  Given a finite metric space $(X,d)$ and a parameter $k ∈ ℕ$,
find a subset of $k$ elements of $X$ that has maximum diversity. There are
many functions that measure diversity. One of the most popular
measures, called \emph{remote-clique}, is the sum of the pairwise distances of the chosen
elements.  In this paper,
we present novel results on three widely used diversity measures:
Remote-clique, remote-star and remote-bipartition.

Our main result are polynomial time approximation schemes for these
three diversification problems under the assumption that the metric space is doubling. This
setting has been discussed in the recent literature. The existence of
such a PTAS however was left open.

Our results also hold in the setting where the distances are raised to a fixed power $q\geq 1$, 
giving rise to more variants of diversity functions, 
similar in spirit to the variations of clustering problems depending on the power applied to the distances. 
Finally, we provide a proof of NP-hardness for remote-clique with squared distances in doubling metric spaces.


 \end{abstract}

\section{Introduction}

A \emph{dispersion} or \emph{diversity maximization} problem is as follows: Given a ground set $X$ and a natural number $k ∈ ℕ$, find a subset $S⊆X$ among those of cardinality $k$ that maximizes a certain \emph{diversity function} $\dv(S)$.  

While diversity maximization has been of interest in the algorithms
and operations research community for some time
already, see e.g.~\cite{chandra2001approximation,birnbaum2009improved,ravi1994heuristic,hassin1997approximation},
the problem received considerable attention in the recent literature regarding
information retrieval, recommender systems, machine learning and data
mining, see e.g.~\cite{vasconcelos2003feature,vieira2011query,qin2012diversifying,radlinski2006improving,abbassi2013diversity}.

Distances used in these applications may be metric (observing the triangle inequality) or non-metric, 
with the latter typically being very challenging for the analysis of heuristics. 
However, most popular distances either are metric or correspond to the $q$-th power of metric distances for some constant $q>1$. 
The cosine distance for a set $X \subseteq \R^d \setminus \{0\}$, for example, 
is a popular non-metric measure of dissimilarity for text documents. The cosine distance between two vectors $x, y \in X \subseteq \R^d \setminus \{0\}$ 
is defined as $d(x, y) = 1 − \cos \theta_{xy}$, where $\theta_{xy}$ is the angle between $x$ and $y$. 
This distance is clearly invariant under positive scalings of $x$ and $y$, 
and it can be checked that if $x$ and $y$ are normalized so that $\|x\| = \|y\| = 1$, 
then their cosine distance is $d(x, y) = \frac{1}{2} \|y-x\|^2$.

\smallskip 
In this paper we focus on three popular diversity functions where the ground set is equipped with a metric distance, 
see e.g. \cite{chandra2001approximation, birnbaum2009improved, aghamolaei2015diversity, indyk2014composable, 
gollapudi2009axiomatic,ceccarello2017mapreduce,borodin2012max,bhaskara2016linear}. 
In particular, for a given $n$-point metric space $(X,d)$, a constant $q\in \R_{\geq 1}$ and a parameter $k\in\Z$ 
with $2\leq k\leq n$, we consider the family of problems 
\[
  \max_{T\subseteq X, |T|=k} \dv^q (T), 
\]
 where $\dv^q(T)$ corresponds to one of the following three diversity functions:
\begin{itemize}
 \item \emph{Remote-clique}: $\displaystyle \cl^q(T):=\sum_{\{u,v\} \in \binom{T}{2}} d^q(u,v)=\frac{1}{2} \sum_{u,v\in T} d^q(u,v)$.
 \item \emph{Remote-star}: $\displaystyle\st^q(T):=\min_{z\in T} \sum_{u\in T\setminus \{z\}} d^q(z,u)$. 
 \item \emph{Remote-bipartition}: $\displaystyle\bp^q(T):=\min_{L\subseteq T, |L|=\lfloor |T|/2 \rfloor} 
\sum_{\ell\in L, r\in T\setminus L} d^q(\ell, r)$.
\end{itemize}
Here, $d^q(u,v)$ is the $q$-th power of the distance between $u$ and $v$. 
In the literature, these problems have been mainly considered for $q=1$ 
to which we refer as \emph{standard} remote-clique, remote-star and remote-bipartition respectively. 
We introduce the study of the generalized versions of these diversification problems 
for any real constant $q\in\R_{\geq 1}$.  
This is in the same spirit as the generalized versions of clustering problems recently introduced in 
\cite{cohen2016local, friggstad2016local}. 

We present PTASs for all three problems in the case where the metric space is \emph{doubling}. Recall that a polynomial time approximation scheme, or PTAS, is an algorithm that, 
for any fixed $\epsilon>0$, returns a $(1-\epsilon)$-approximation in polynomial time. The latter is a general and robust class of metric spaces that have low intrinsic dimension and often occur in applications. 
We provide a proper definition in Section~\ref{s:Prelim}.

\subsection*{Contributions of this paper}

Suppose that $(X,d)$ is a metric space of bounded doubling dimension $D$ and that the power $q\geq 1$ is fixed. 
In this setting, our main results are as follows:

\begin{enumerate}
\item We show that there exist polynomial time approximation schemes (PTAS)
  for the remote-clique, remote-star and remote-bipartition problems. 
	We prove this result by means of a single and very simple algorithm, built upon two results: 
	\begin{enumerate}
	\item We show a structural property of all these diversification problems whereby each instance can be partitioned into 
	a main cluster of points of bounded diameter and its complement which has to be part of the optimal solution. 
	\item We perform grid rounding (a standard technique for geometric problems) restricted to the above-mentioned cluster. 
	For this, we show how the analysis of grid rounding can be easily generalized to the $q$-th power of metric distances.
	\end{enumerate}
	We consider both of these results to be of independent interest.

\item For the standard ($q=1$) remote-clique problem, which is arguably the most popular of the problems in consideration, 
we refine our generic algorithm into a fast PTAS that runs in time 
$ O(n(k+\epsilon^{-D}))+ (\epsilon^{-1} \log k)^{O(\epsilon^{-D})}\cdot k$. 
Notice the (optimal) linear dependence on $n$ and mild dependence on $k$, which makes it implementable for very large instances. 
We also remark that even for standard remote-clique, ours is the first PTAS in the literature for general doubling metrics. 
 
\item For the remote-bipartition problem, our algorithm assumes access to a polynomial time oracle that, 
for any $k$-set $T$, returns the value of $\bp^q (T)$. 
For $q=1$, this corresponds to the metric min-bisection problem, known to be NP-hard and admitting a PTAS \cite{fernandez2002polynomial}. 
We generalize this last result and provide a PTAS for min-bisection over doubling metric spaces 
for \emph{any} constant $q \geq 1$, thus validating our main result. 

\item We provide the first NP-hardness proof for remote-clique in fixed doubling dimension. 
More precisely, we prove that the version of remote-clique with squared Euclidean distances in $\R^3$ is NP-hard. 
\end{enumerate}

\begin{table}[t]
\centering
\resizebox{\textwidth}{!}{%
\begin{tabular}{@{}lccccc@{}}
\toprule
\multicolumn{1}{c}{\multirow{2}{*}{\textbf{Problem}}} 
     & \multicolumn{1}{c}{\multirow{2}{*}{\textbf{\begin{tabular}[c]{@{}c@{}}Distance\\ class\end{tabular}}}} 
     & \multicolumn{2}{c}{\textbf{Unbounded dimension}}  & \multicolumn{2}{c}{\textbf{Fixed (doubling) dimension}} \\
\multicolumn{1}{c}{}  & \multicolumn{1}{c}{} & \textbf{Approx.} & \textbf{Hardness} & \textbf{Approx.} & \textbf{Hardness} \\ \midrule
\multirow{2}{*}{clique, $q=1$} & Metric & $1/2$ \cite{hassin1997approximation, birnbaum2009improved} 
     & $1/2 + \epsilon \ \dagger$ \cite{borodin2012max} & PTAS (Thm.~\ref{thm:PTAS1}) & -- \\ \cmidrule(l){2-6}
& $\ell_1$, $\ell_2$ & PTAS \cite{cevallos_2016_max-sum, cevallos2017local} & NP-hard \cite{cevallos_2016_max-sum} 
     & PTAS \cite{fekete2004maximum, cevallos_2016_max-sum, cevallos2017local} & -- \\ \midrule
clique, $q=2$ & $\ell_2$ & PTAS \cite{cevallos_2016_max-sum, cevallos2017local} & NP-hard \cite{cevallos_2016_max-sum} 
     & PTAS \cite{cevallos_2016_max-sum, cevallos2017local} & NP-hard (Thm.~\ref{thm:NP}) \\ \midrule
star, $q=1$ & Metric & 1/2 \cite{chandra2001approximation} & $1/2 + \epsilon \ \dagger$ (Thm.~\ref{thm:planted}) 
     & PTAS (Thm.~\ref{thm:PTAS1}) & -- \\ \midrule 
bipartition, $q=1$ & Metric & 1/3 \cite{chandra2001approximation} & $1/2 + \epsilon \ \dagger$ (Thm.~\ref{thm:planted}) 
     & PTAS (Thm.~\ref{thm:PTAS1}) & -- \\ \midrule
\begin{tabular}[c]{@{}l@{}}3 problems, any\\ const. $q \geq 1$ \end{tabular} & Metric 
     & -- & $2^{-q} + \epsilon \ \dagger$ (Thm.~\ref{thm:planted}) & PTAS (Thm.~\ref{thm:PTAS1}) & NP-hard (Thm.~\ref{thm:NP}) \\ \bottomrule
\end{tabular}%
}
\caption{Current best approximation ratios and hardness results for remote-clique, remote-star and remote-bipartition 
with a highlight on our results. The sign $\dagger$ indicated that the result assumes hardness of the \emph{planted-clique problem}.}
\label{tab:stateoftheart}
\end{table}





\subsection*{Related work}

For the standard case $q=1$ and for general metrics, Chandra and Halld\'orsson~\cite{chandra2001approximation} 
provided a thorough study of several diversity problems, including remote-clique, remote-star and remote-bipartition. 
They observed that all three problems are NP-hard by reductions from the CLIQUE-problem and provided a $\frac{1}{2}$-factor 
and a $\frac{1}{3}$-factor approximation algorithm for remote-star and remote-bipartition respectively. 
Several approximation algorithms are known for remote-clique as well 
\cite{ravi1994heuristic, hassin1997approximation, birnbaum2009improved} with the current best factor being $\frac{1}{2}$. 

\begin{remark}
 Borodin et al.~\cite{borodin2012max} proved that the approximation factor of $\frac{1}{2}$ 
is best possible for standard remote-clique over general metrics 
 under the assumption that the \emph{planted-clique problem}~\cite{alon2011inapproximability} is hard. 
We prove in Theorem~\ref{thm:planted} of Section~\ref{s:NP} that, under the same assumption and for any $q\geq 1$, 
 neither remote-clique, remote-star nor remote-bipartition admits a constant approximation factor higher than $2^{-q}$. 
 Thus, none of the three problems nor their generalizations for $q\geq 1$ admits a PTAS over general metrics.  
\end{remark}

In terms of relevant special cases for standard remote-clique, Ravi et al.~\cite{ravi1994heuristic} provided an efficient exact 
algorithm for instances over the real line, and a factor of $\frac{2}{\pi}$ over the Euclidean plane.  
Later on, Fekete and Meijer~\cite{fekete2004maximum} provided the first PTAS for this problem for fixed-dimensional $\ell_1$ distances, and an improved factor of $\frac{\sqrt{2}}{2}$ over the Euclidean plane.
Very recently, Cevallos et al.~\cite{cevallos_2016_max-sum, cevallos2017local} provided PTASs over $\ell_1$ and $\ell_2$ distances 
of unbounded dimension as well as for distances of \emph{negative type}, a class that contains some popular non-metric distances including the cosine distance. 
We remark however that the running times of all previously mentioned PTASs 
\cite{fekete2004maximum, cevallos_2016_max-sum, cevallos2017local} have a dependence on $n$ given by high-degree polynomials 
(in the worst case) and thus are not suited for large data sets. 

For remote-star and remote-bipartition, to the best of the authors' knowledge there were no previous results in the literature 
on improved approximability for any fixed-dimensional setting, nor for other non-trivial special settings beyond general metrics. 
Moreover, there was no proof of NP-hardness for any of the three problems in a fixed-dimensional setting. 
In particular, showing NP-hardness of a fixed-dimensional version of remote-clique 
was left as an open problem in~\cite{fekete2004maximum}.

\subsection*{Further related results and implications}

In applications of diversity maximization in the area of information retrieval, 
common challenges come from the fact that the data sets are very large 
and/or are naturally embedded in a high dimensional vector space. 
There is active research in dimensionality reduction techniques, see \cite{cunningham2015linear} for a survey. 
It has also been remarked that in many scenarios such as human motion data and face recognition, 
data points have a hidden intrinsic dimension that is very low and independent from the ambient dimension, 
and there are ongoing efforts to develop algorithms and data structures that exploit this fact, see \cite{tenenbaum2000global, indyk2007nearest, dasgupta2008random, gottlieb2015nonlinear}. 
One of the most common and theoretically robust notions of intrinsic dimension is precisely the doubling dimension.
We remark that our algorithm does not need to embed the input points into a vector space (of low dimension or otherwise) 
and does not require knowledge of the doubling dimension, as this parameter only plays a role in the run-time analysis. 

Our proposed generalization of the diversity problems into powers of metric distances is relevant for heuristics related to 
\emph{snowflake metrics}, which have recently gained attention in the literature of dimensionality reduction, 
see \cite{Bartal:2011:DRB:2133036.2133104, gottlieb2015nonlinear, Bartal_2016_dimension}. 
A snowflake of a metric $(X,d)$ is the metric $(X, d^\alpha)$, 
where each distance has been raised to the power $\alpha$ for some constant $0<\alpha<1$, 
and this new metric space often admits embeddings into much lower-dimensional vector spaces than $(X,d)$ does. 
As an application example, Bartal and Gottlieb~\cite{Bartal_2016_dimension} recently presented 
an efficient approximation algorithm for a clustering problem (called min-sum) that works as follows: 
Given an input metric $(X,d)$, they consider its snowflake $(X,d^{1/2})$, 
embed it into a low-dimensional Euclidean space (incurring low distortion), 
and then solve the original instance by applying over the new instance 
a readily available algorithm for the squared distances version of the problem.

A sensible approach when dealing with very large data sets is to perform a \emph{core-set reduction} of the input as a pre-processing step. 
This procedure quickly filters through the input points and discards most of them, 
leaving only a small subset -- the core-set -- that is guaranteed to contain a near-optimal solution. 
There are several recent results on core-set reductions for standard ($q=1$) dispersion problems, see \cite{indyk2014composable, aghamolaei2015diversity, ceccarello2017mapreduce}. 
In particular, Ceccarello et al.~\cite{ceccarello2017mapreduce} recently presented a PTAS-preserving reduction (resulting in an arbitrarily small deterioration of the approximation factor) for all three problems in doubling metric spaces, with the existence of a PTAS left open.
Their construction allows for our algorithm to run in a machine of restricted memory and adapts it to streaming and distributed models of computation. Besides showing that a PTAS exists, we can combine our results with theirs. We refer the interested reader to the previously mentioned references and limit ourselves to remark a direct consequence 
of Theorem~\ref{thm:PTAS1} and \cite[Theorems 3 and 9]{ceccarello2017mapreduce}.

\begin{corollary}
\label{co:1}
 For $q=1$ and any constant $\epsilon>0$, our three diversity problems over metric spaces of constant doubling dimension $D$ admit $(1-\epsilon)$-approximations that execute as single-pass and 2-pass streaming algorithms, 
 in space $O(\epsilon^{-D} k^2)$ and $O(\epsilon^{-D} k)$ respectively.
\end{corollary}

\subparagraph*{Organization of the paper.} 
In Section~\ref{s:Prelim} we provide some needed notation, background techniques and results.
Section~\ref{s:PTAS1} presents the general algorithm (Theorem~\ref{thm:PTAS1}) and Section~\ref{s:PTAS2} describes a faster algorithm for standard remote-clique (Theorem~\ref{thm:PTAS2}). 
Section~\ref{s:BP} contains a PTAS for the generalized min-bisection problem (Theorem~\ref{thm:BP}) 
and Section~\ref{s:NP} is dedicated to hardness results (Theorems \ref{thm:planted} and \ref{thm:NP}). Acknowledgements can be found in Section~\ref{s:Ack}. Finally, for best readability the proofs of some lemmas have been moved to Appendix~\ref{s:proofs}. 


\section{Preliminaries}\label{s:Prelim}

A \emph{(finite) metric space} is a tuple $(X,d)$, where $X$ is a finite set and $d:X\times X\rightarrow \R_{\geq 0}$ 
is a symmetric distance function that satisfies the triangle inequality with $d(u,u)=0$ for each point $u\in X$. 
For any set $S\subseteq X$, its complement is denoted by $\bar{S}:= X\setminus S$. 
For a point $u\in X$ and a parameter $r\in \R_{\geq 0}$, the \emph{ball centered} at $u$ of radius $r$ is defined as 
$B(u,r):=\{v\in X: \ d(u,v)\leq r\}$. 
The \emph{doubling dimension} of $(X,d)$ is the smallest $D\in \R_{\geq 0}$ 
such that any ball in $X$ can be covered by at most $2^D$ balls of half its radius. In other words, for each $u\in X$ and $r> 0$, there exist points $v_1, \cdots, v_t\in X$ with $t\leq 2^D$ 
such that $B(u,r)\subseteq \cup_{i=1}^t B(v_i, r/2)$. 
A family of metric spaces is \emph{doubling} if their doubling dimensions are bounded by a constant.  
It is well known that all metric spaces induced by a normed vector space of bounded dimension are doubling.

We rely on the standard \emph{cell decomposition} technique and \emph{grid rounding}, see~\cite{har2011geometric}.
We assume without loss of generality that the diameter of $(X,d)$ (the largest distance between two points) is one. 
For a parameter $\delta>0$, the following greedy procedure partitions $X$ into \emph{cells} of radius $\delta$. 
Initially, define all points in $X$ to be white. While there exist white points, pick one that we call $u$, 
color it red and assign all white points $v\in X$ with $d(u,v)\leq \delta$ to $u$ and color them blue. 
A cell is now comprised of a red point, declared to be the cell center, and all the blue points assigned to it. 
Notice that this procedure executes in time $O((\text{\# cells})\cdot |X|)$ 
and requires no knowledge of the value of the doubling dimension $D$. 
For any set $S\subseteq X$, we denote by $\pi(S)$ the set of centers of those cells that intersect with $S$. In other words, subset $\pi(X)$ is a \emph{$\delta$-net} of $X$, 
meaning that a) any two points in $\pi(X)$ are at distance strictly larger than $\delta$, 
and b) for any point in $X$ there is a point in $\pi(X)$ at distance at most $\delta$. We denote by $\hat{\pi}(S)$ the corresponding multiset over $\pi(S)$ with each point $u\in \pi(S)$ 
having multiplicity $|\pi^{-1}(u)\cap S|$, see Figure~\ref{fig:decomposition}.

 \begin{figure}[t]
   \centering
   \includegraphics[height=2.5cm]{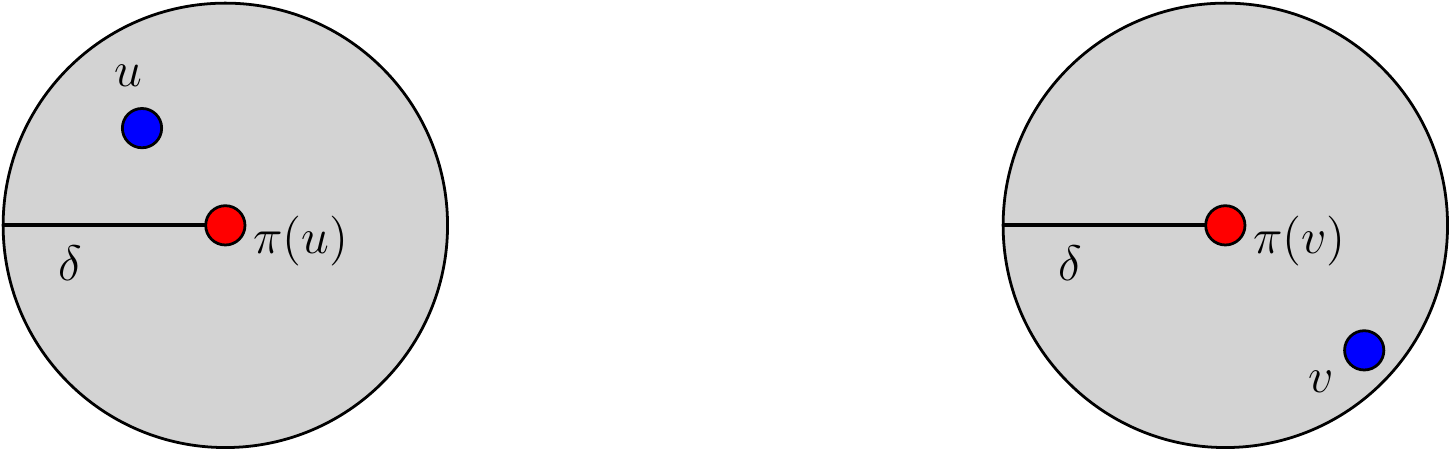}
   \caption{A cell decomposition of radius $\delta$. For each point $u\in X$, one has $d(u, \pi(u))\leq \delta$.}
   \label{fig:decomposition}
 \end{figure}

Grid rounding means to \emph{move} or \emph{round} each point to its respective cell center, to turn $X$ into the multiset $\hat{\pi}(X)$. 
This incurs a location error of at most $δ$ for each point. 
How many cells and thus distinct positions does this algorithm produce? 
If $(X,d)$ is of doubling dimension $D$, a direct consequence of the definition is that for any parameters 
$r$ and $\rho$ in $\R_{> 0}$, a ball of radius $r$ can be covered by at most $(2/\rho)^D$ balls of radius $\rho r$. 
Since $X$ is contained in a ball of radius one, the number of cells produced is bounded by $(4 /\delta)^D$. 
Indeed, $X$ can be covered by $(4 /\delta)^D$ balls of radius $\delta/2$ and each such ball contains at most one cell center since, 
by construction, the distance between any two cell centers is strictly larger than $\delta$. 
If parameters $\delta$ and $D$ are considered constant, then the number of distinct positions in the rounded instance $\hat{\pi}(X)$ 
is bounded by a constant. An optimal solution to many geometric problems can be found efficiently by exhaustive search. For instances of diversity maximization over the multiset $\hat{\pi}(X)$,  
we consider each high-multiplicity point as distinct points having pairwise distance zero.

\smallskip

The following two lemmas correspond respectively to standard inequalities used for powers of metric distances, see \cite{cohen2016local, friggstad2016local}, and to trivial relations among our three diversity functions. 
Their proofs are deferred to Appendix~\ref{s:proofs}.

\begin{lemma}\label{lem:relaxedtriangle}
Fix a constant $q\geq 1$. For any three points $u,v,w\in X$ one has 
 \begin{align}
  d^q(u,w) &\leq 2^{q-1} \big[d^q(u,v)+d^q(v,w)\big] \ \text{ or equivalently }  \label{eq:triangle1} \\ 
  \ d^q(u,v) &\geq 2^{-(q-1)} d^q(u,w) - d^q(v,w). \label{eq:triangle2}
 \end{align}
For any numbers $x,y\in \R_{\geq 0}$ and $0\leq \epsilon \leq 1$,  
\begin{align}
 (x+\epsilon y)^q \leq x^q + 2^q \epsilon \cdot \max\{x^q,y^q\}. \label{eq:triangle3}
\end{align}
\end{lemma}

\begin{lemma}\label{lem:objectiverelations}
 Fix a constant $q\geq 1$. For any $k$-set $T\subseteq X$,
 \begin{align}
  \frac{k}{2} \cdot \st^q(T)  &\leq \cl^q(T) \leq 2^{q-1} k \cdot \st^q(T) \quad \text{and} \label{eq:relation1}\\
   \frac{2 (k-1)}{k}\cdot \bp^q(T) &\leq \cl^q(T) \leq (2^q+1) \cdot \bp^q(T) \quad \text{(assuming that $k$ is even).}  \label{eq:relation2}
 \end{align}
\end{lemma}

Whenever we deal with remote-bipartition, we assume for simplicity that $k$ is even 
-- all our results can easily be extended to the odd case, up to a change in constants by a factor $2^{O(q)}$. 
Therefore, the diversity functions correspond to the sum of $\binom{k}{2}$, $(k-1)$ and $k^2/4$ terms, 
respectively for remote-clique, remote-star and remote-bipartition. 
Consequently, for each function $\dv^q$ and for a given instance, we fix an optimal $k$-set denoted by $OPT_{\dv^q}$  
and define its \emph{average optimal value} $\Delta_{\dv^q}$ as follows:
\begin{itemize}
 \item $\displaystyle \Delta_{\cl^q}:= \cl^q(OPT_{\cl^q})/ \binom{k}{2}$,
 \item $\displaystyle \Delta_{\st^q}:= \st^q(OPT_{\st^q}) / (k-1)$,
 \item $\displaystyle\Delta_{\bp^q}:= \bp^q(OPT_{\bp^q}) / (k^2/4)$.
\end{itemize}

Whenever the diversity function $\dv^q$ is clear from context, or for general statements on all three functions, we use $OPT$ and $\Delta$ as short-hands for $OPT_{\dv^q}$ and $\Delta_{\dv^q}$ respectively. 

\begin{remark}
 It directly follows from Lemma~\ref{lem:objectiverelations} that for a common metric space and common parameters $q\geq 1$ and $k$, 
 the average optimal values $\Delta_{\cl^q}$, $\Delta_{\st^q}$ and $\Delta_{\bp^q}$ are all just a constant away from each other 
 (a constant $2^{O(q)}$ that is independent of $n$ and $k$). 
 We heavily use this property linking our three problems in the proof of our key structural result (Theorem \ref{thm:far=opt}). 
 A similar result does not extend to other common diversity maximization problems such as remote-edge, remote-tree and remote-cycle,  see~\cite{chandra2001approximation} for definitions. 
 This seems to be a bottleneck for possibly adapting our approach to those problems.
\end{remark}


\section{A PTAS for all three diversity problems}\label{s:PTAS1}

We now come to our main result which is the following theorem. 

\begin{theorem} \label{thm:PTAS1}
For any constant $q \in \mathbb{R}_{\geq 1}$, the $q$-th power versions of the remote-clique, remote-star and remote-bipartion problems admit PTASs over doubling metric spaces.
\end{theorem}

Let us fix a constant error parameter $\epsilon>0$. 
Our algorithm is based on grid rounding. 
However, if we think about the case $q=1$, a direct implementation of this technique requires a cell decomposition 
of radius $O(\epsilon \cdot \Delta)$, which is manageable only if $\Delta$ is large enough with respect to the diameter. 
Otherwise, the number of cells produced may be super-constant in $n$ or $k$. 
Hence, a difficult instance is one where $\Delta$ is very small, which intuitively occurs only in the degenerate case where 
most of the input points are densely clustered in a small region, with very few points outside of it, see Figure~\ref{fig:twostars}.
The algorithmic idea is thus to partition the input points into a \emph{main cluster} and a collection of \emph{outliers}, 
and treat these sets differently.

\begin{figure}[t]
\centering
\begin{tabular}{lr}
\includegraphics[scale=0.45]{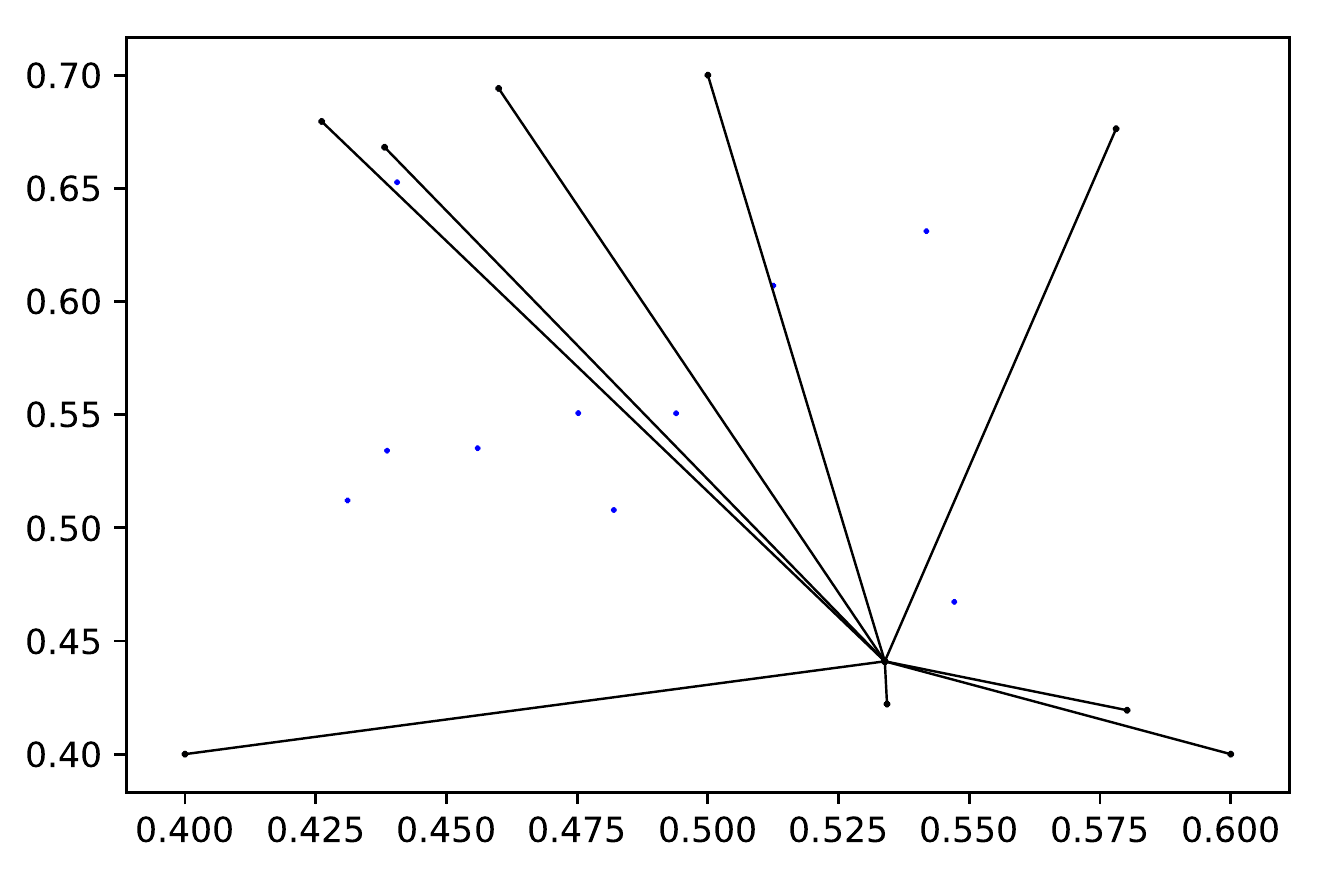}
&
\includegraphics[scale=0.45]{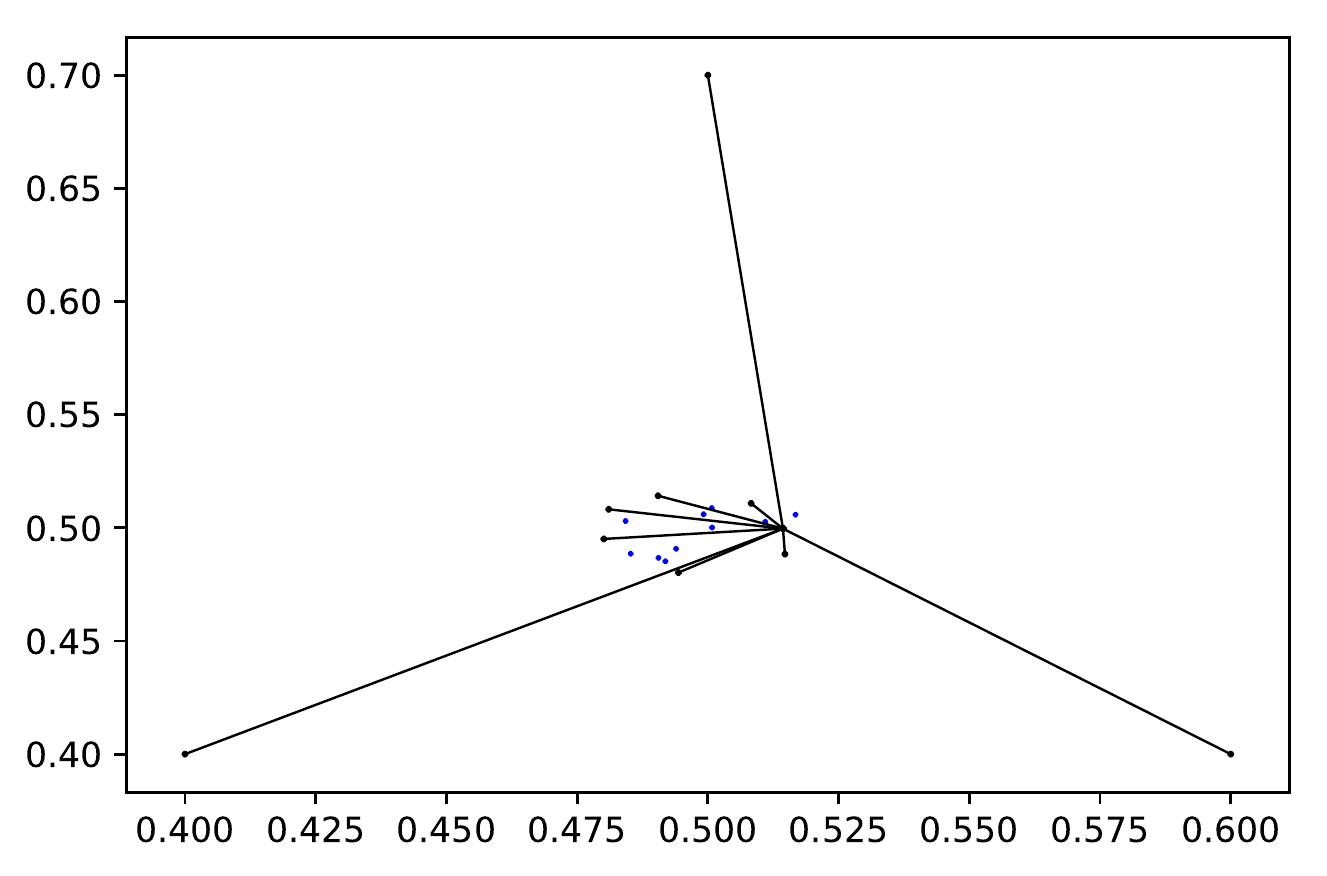}
\end{tabular}
\caption{Two instances of remote-star over the Euclidean plane, with $q=1$ and $k=10$, 
drawn together with the minimum weight spanning star of the optimal solution $\opt$. 
On the left, a well-dispersed instance where $\Delta$ is of the same order as the diameter. 
On the right, a degenerate instance with small $\Delta$ 
and a clear partition of the input into a main cluster and a set of three outliers.}
\label{fig:twostars}
\end{figure}

\subsection{Key structural result}\label{ss:mainthm}

We identify in any instance a main cluster containing most of the input points. 
This cluster corresponds to a ball with a radius that is bounded with respect to $\Delta^{1/q}$. 
Thanks to the nature of the diversity functions, we can guarantee that \emph{all outliers are contained in $\opt$}. 

\begin{theorem}\label{thm:far=opt}
Fix a constant $q\geq 1$. 
For each diversity function $\dv^q$ in $\{\cl^q, \st^q, \bp^q\}$ and a fixed optimal $k$-set $\opt_{\dv^q}\subseteq X$, 
there is a point $z_0 = z_0(\dv^q)$ in $\opt_{\dv^q}$ so that 
\[
X\setminus B(z_0,c_{\dv^q} (\Delta_{\dv^q})^{1/q}) \subseteq \opt_{\dv^q},
\]
where $c_{\cl^q}=2$, $c_{\st^q}=4$, and $c_{\bp^q}=6$.
\end{theorem}

\begin{proof}
For each function $\dv^q$ in $\{\cl^q, \st^q, \bp^q\}$, let $z_0=z_0(\dv^q)$ be the center of the minimum weight spanning star in $\opt_{\dv^q}$ so that $\st^q(\opt_{\dv^q}) = \sum_{u\in \opt_{\dv^q}} d^q(z_0,u)$. 
Consider a point $s=s(\dv^q)$ outside of the ball $B(z_0,c_{\dv^q} (\Delta_{\dv^q})^{1/q})$, i.e.
\begin{align}
d^q(z_0, s) > (c_{\dv^q})^q \cdot \Delta_{\dv^q}. \label{eq:s=far}
\end{align}
 Assume that $s$ is not in $\opt_{\dv^q}$ and define the $k$-set $\opt'_{\dv^q}:=\opt_{\dv^q}\cup\{s\}\setminus \{z_0\}$. 
We will show for each diversity function that $\dv^q(\opt'_{\dv^q})> \dv^q(\opt_{\dv^q})$, thus contradicting the optimality of $\opt_{\dv^q}$. To simplify notation in the remainder of the proof, we make the corresponding function clear from context and 
 remove the subscripts $\dv^q$. 
 
For remote-clique, we have
 \begin{align*}
  \cl^q(\opt')-\cl^q(\opt) & = \sum_{u\in \opt\setminus \{z_0\}} \big[d^q(s,u)-d^q(z_0,u)\big] \\
			& \geq \sum_{u\in \opt\setminus \{z_0\}} \big[2^{-(q-1)} d^q(z_0,s) - 2d^q(z_0,u)\big]  & \text{(by~\eqref{eq:triangle2})}\\
			& = \frac{k-1}{2^{q-1}} d^q(z_0,s) - 2 \cdot \st^q(\opt) & \text{(by choice of $z_0$)} \\
		       & > \frac{k-1}{2^{q-1}} (2^q \Delta) - 2\cdot \frac{2}{k} \cdot \cl^q(\opt)	& \text{(by \eqref{eq:s=far} and \eqref{eq:relation1})} \\
		       & = 2(k-1)\Delta - 2(k-1)\Delta = 0 & \text{(by def.~of $\Delta$)}.
 \end{align*}

For remote-star, let $z$ be the center of the minimum weight spanning star in $\opt'$ so that $\st^q(\opt')=d^q(z, s) + \sum_{u\in \opt\setminus\{z_0\}} d^q(z,u)$. We claim that 
\begin{align}
d^q(z_0, z) \leq 2^q \Delta, \label{eq:z=close}
\end{align} 
as otherwise we obtain
 \begin{align*}
  \st^q(\opt)+\st^q(\opt') & = d^q(z, s) + \sum_{u\in \opt\setminus\{z_0\}} \big[ d^q(z_0,u) + d^q(z,u) \big] \\
			&\geq 2^{-(q-1)} \sum_{u\in\opt\setminus\{z_0\}} d^q(z_0,z) & \text{(by~\eqref{eq:triangle1})} \\
			& > \frac{k-1}{2^{q-1}} (2^q \Delta) = 2(k-1) \Delta =  2 \cdot \st^q(\opt) & \text{(negating \eqref{eq:z=close}).}
 \end{align*}
Inequality~\eqref{eq:z=close} implies in particular that $z\neq s$, hence $z\in\opt$. 
Notice by the minimality of the remote-star function that $\st^q(\opt)\leq \sum_{u\in \opt} d^q(z,u)$. 
By inequalities \eqref{eq:triangle2}, \eqref{eq:s=far} and \eqref{eq:z=close}, we obtain
\begin{align*}
 \st^q(\opt') - \st^q(\opt) & \geq \sum_{u\in \opt'} d^q(z,u) - \sum_{u\in \opt} d^q(z,u) = d^q(z,s) - d^q(z,z_0) \\
			& \geq 2^{-(q-1)} d^q (z_0,s) - 2d^q(z_0,z) \\
                               & > 2^{-(q-1)} (4^q \Delta) - 2 (2^q \Delta) =0. 
\end{align*}

For remote-bipartition, let $\opt'=L'\cup R$ be the minimum weight bipartition of $\opt'$ so that $\bp^q(\opt') = \sum_{\ell\in L', r\in R} d^q(\ell,r)$. Assume without loss of generality that $s\in L'$. We claim that 
\begin{align}
\sum_{r\in R} d^q(z_0,r) \leq \frac{2^q+1}{2} k\Delta, \label{eq:R=close}
\end{align}  
as otherwise we obtain 
\begin{align*}
 \bp^q(\opt) & \geq \frac{1}{2^q+1} \cl^q(\opt) \geq \frac{k}{2(2^q+1)} \st^q(\opt) & \text{(by \eqref{eq:relation2} and \eqref{eq:relation1})}\\
	    &= \frac{k}{2(2^q+1)} \sum_{u\in \opt} d^q(z_0,u) \geq \frac{k}{2(2^q+1)} \sum_{r\in R} d^q(z_0,r)  & \text{(as $R\subseteq \opt$)}\\
	    & > \frac{k}{2(2^q+1)} \cdot \frac{2^q+1}{2} k\Delta = \frac{k^2}{4} \Delta = \bp^q(\opt) & \text{(negating \eqref{eq:R=close}).}
\end{align*}
Define $L:=L'\cup \{z_0\}\setminus \{s\}$ and notice that $L\cup R=\opt$. By the minimality of the remote-bipartition function, 
$\bp^q(\opt)\leq \sum_{\ell \in L} \sum_{r\in R} d^q(\ell,r)$. Hence,
\begin{align*}
\bp^q(\opt') - \bp^q(\opt) &\geq  \sum_{\ell \in L'} \sum_{r\in R} d^q(\ell,r) - \sum_{\ell \in L} \sum_{r\in R} d^q(\ell,r) \\
			  &= \sum_{r\in R} \big[ d^q(s,r) - d^q(z_0,r) \big] \\
			  &\geq \sum_{r\in R} \big[ 2^{-(q-1)} d^q (z_0,s) - 2 d^q(z_0,r) \big]  & \text{(by~\eqref{eq:triangle2})} \\
			  &> \frac{|R|}{2^{q-1}} (6^q \Delta) - 2\sum_{r\in R} d^q(z_0,r) & \text{(by \eqref{eq:s=far})} \\
			& \geq 3^q k \cdot \Delta  - (2^q+1)k \cdot \Delta \geq 0. & \text{(by \eqref{eq:R=close})}.
\end{align*}
This completes the proof of the theorem.
\end{proof}

\subsection{The algorithm} \label{ss:algorithm}

For any diversity function and a fixed optimal $k$-set, we refer to the ball $B:=B(z_0,c \Delta^{1/q})$ defined in 
Theorem~\ref{thm:far=opt} as the \emph{main cluster} and to $z_0$ as the \emph{instance center}. 
Our algorithm consists of two phases: Finding the main cluster $B$ and performing grid rounding on $B$. 
We remark that for a well-dispersed instance (such as the one on the left-hand side of Figure~\ref{fig:twostars}), 
$B$ may well contain all input points. In that case, our algorithm amounts to a direct application of the grid rounding procedure.  

\subsubsection*{Finding the main cluster}\label{ss:cluster}

There are several possible ways to (approximately) find $B$. For simplicity, we present a naive approach based on exhaustive search. 
A more efficient technique is described in Section~\ref{s:PTAS2}, where we provide a faster algorithm for standard remote-clique. 

Assuming without loss of generality that the instance diameter is 1, we obtain for each diversity function the bounds 
$1/k^2 \leq \Delta^{1/q} \leq 1$. Hence, by performing $O(\log k)$ trials, we can ``guess'' the value of $\Delta^{1/q}$ 
up to a constant factor arbitrarily close to one, which means that for any constant $\lambda>0$, 
we can find an estimate $\Delta'$ so that $(1-\lambda) \Delta^{1/q}\leq \Delta'^{1/q}\leq \Delta^{1/q}$.   
Similarly, by trying out all $n$ input points, we can ``guess'' the instance center $z_0$. 
For each one of these guesses, we perform the second phase (described in the next paragraph) 
and output the best $k$-set found over all trials. 
To simplify our exposition, we assume in what follows that we have found $\Delta^{1/q}$ and $z_0$ (and thus $B$) exactly. 
Our analysis can be adapted to any constant-factor estimation of $\Delta^{1/q}$,  
as it is enough to find a slightly larger ball $B'$ containing $B$ and to slightly change the value of constant $c$. More precisely, if we have an estimate $\Delta'$ so that $(1-\lambda) \Delta^{1/q}\leq \Delta'^{1/q}\leq \Delta^{1/q}$ 
and we set $c':=\frac{c}{1-\lambda}$, then $B':=B(z_0, c' \Delta'^{1/q})$ is guaranteed to contain $B$ 
and hence all points outside of $B'$ are in $\opt$.

\subsubsection*{Rounding the cluster}\label{ss:rounding}

We now assume that we have found the main cluster $B$ (see the previous paragraph).
For a constant $\delta>0$ to be defined later, with $1/\delta=\Theta(2^q/\epsilon)$, 
we perform a cell decomposition of radius $\delta \Delta^{1/q}$ over $B$. 
As the radius of ball $B$ is $c\Delta^{1/q}$, this decomposition produces at most 
$\big(4\cdot \frac{c \Delta^{1/q}}{\delta \Delta^{1/q}}\big)^D= (4c / \delta)^D = O(2^q/\epsilon)^D$ cells, i.e. constantly many cells. 
Let $\pi:B\rightarrow B$ be the function that maps each point to its cell center. 
For notational convenience, we extend this into a function $\pi:X\rightarrow X$ by applying the identity on $X\setminus B=:\bar{B}$  
(and thinking of each point in $\bar{B}$ as the center of its own cell). 

Recall that for any set $T\subseteq X$, $\hat{\pi}(T)$ denotes the multiset over set 
$\pi(T)$ having multiplicities $|\pi^{-1}(u)\cap T|$ for each $u\in \pi(T)$. Next, we perform exhaustive search to find a $k$-set $T$ in $X$ with the property that 
\begin{align}\label{eq:roundedOPT}
\dv^q(\hat{\pi}(T))\geq \dv^q(\hat{\pi}(\opt)).
\end{align}

This can be done in polynomial time as follows: We try out all multisets in $\hat{\pi}(X)$ that 
a) contain $\bar{B}$ and b) have cardinality $k$ counting multiplicities. 
Then, we keep the multiset with largest diversity and return any $k$-set $T$ that is a pre-image of this multiset. 
Clearly, this search considers only $k^{O(2^q/\epsilon)^D}$ multisets and is bound to consider $\hat{\pi}(\opt)$.

As mentioned in the introduction, our algorithm assumes access to a polynomial-time oracle that, for any $k$-set $T$, 
returns the value of $\dv^q(T)$ or a $(1+\epsilon)$-factor estimate of it which is sufficient for our purposes. 
The use of this estimate produces a corresponding small deterioration in our final approximation guarantee, 
but for simplicity we ignore this in the remainder.
No exact efficient algorithm is known to compute $\bp^q(T)$ for a given $k$-set $T$. 
However, we provide a PTAS for this problem in Section~\ref{s:BP}.

\subsection{Analysis}\label{ss:analysis}

What is the approximation guarantee of our algorithm? 
We now show how the analysis of the grid rounding heuristic can be adapted to the $q$-th power of metric distances. 
By an application of inequality~\eqref{eq:triangle3}, our cell decomposition gives the following guarantee for each pair of points.

\begin{lemma}\label{lem:mapping}
 Let $\pi:X\rightarrow X$ be a map such that $d(u,\pi(u))\leq \delta \Delta^{1/q}$ for each $u$ in $X$. Then, for any pair of points $u,v\in X$,
 \begin{align*}
  |d^q(u,v) - d^q(\pi(u), \pi(v))| \leq 2^{q+1} \delta \cdot \big(\Delta + \min\{d^q(u,v), d^q(\pi(u), \pi(v))\} \big). 
 \end{align*}
\end{lemma}

\begin{proof}
 We consider two cases. If $d(u,v)\leq d(\pi(u), \pi(v))$, we have by hypothesis 
 \begin{align*}
  d^q(\pi(u),\pi(v)) & \leq [d(\pi(u),u)+d(u,v)+d(v,\pi(v))]^q \leq [d(u,v)+ 2\delta \Delta^{1/q}]^q \\
	   & \leq d^q(u,v) + 2^{q+1}\delta \cdot \max\{\Delta, d^q(u,v)\} \leq d^q(u,v) + 2^{q+1}\delta \cdot \big( \Delta + d^q(u,v)\big), 
 \end{align*}
 where we used inequality~\eqref{eq:triangle3} in the second line. This proves the claim.

 Similarly, if $d(\pi(u), \pi(v)) < d(u,v)$, then
 \[
  d^q(u,v) \leq d^q(\pi(u),\pi(v)) + 2^{q+1}\delta \cdot \big(\Delta + d^q(\pi(u),\pi(v)) \big),
 \]
which again proves the claim.
\end{proof}

Lemma \ref{lem:mapping}, together with the definition of $\Delta$, implies the following result 
whose proof is deferred to Appendix~\ref{s:proofs}.

\begin{lemma}\label{lem:roundedset}
 Let $\pi:X\rightarrow X$ be a map such that $d(u,\pi(u))\leq \delta \Delta^{1/q}$ for each $u$ in $X$. 
 Then, for each one of our three diversity functions and for each $k$-set $T\subseteq X$,
 \begin{align*}
  |\dv^q(T) - \dv^q(\hat{\pi}(T))| \leq 2^{q+1}\delta \cdot \big[ \dv^q(\opt)+\dv^q(T) \big] \leq 2^{q+2}\delta \cdot \dv^q(\opt).
 \end{align*}
\end{lemma}
Applying the previous lemma twice as well as inequality~\eqref{eq:roundedOPT} once, we conclude that 
\begin{align*}
 \dv^q(T) &\geq \dv^q(\hat{\pi}(T)) - 2^{q+2} \delta\cdot \dv^q(\opt) \geq \dv^q(\hat{\pi}(\opt)) - 2^{q+2} \delta\cdot \dv^q(\opt) \\
	  &\geq \dv^q(\opt) - 2^{q+3} \delta\cdot \dv^q(\opt) = (1-2^{q+3}\delta )\cdot \dv^q(\opt). 
\end{align*}
Hence, in order to achieve an approximation factor of $1-\epsilon$, it suffices to select $\delta:=\epsilon / 2^{q+3}$. 
The number of cells produced by the cell decomposition is thus bounded by $(2^{q+5}c/\epsilon)^D = O(2^q/\epsilon)^D$.
This completes the analysis of our algorithm and the proof of Theorem~\ref{thm:PTAS1}.

\section{A faster PTAS for standard remote-clique}\label{s:PTAS2}

We next describe a faster algorithm for the standard ($q=1$) remote-clique problem and prove the following result. 
\begin{theorem}\label{thm:PTAS2}
For any constant $\epsilon>0$, standard remote-clique over an $n$-point metric space of constant doubling dimension $D$ 
admits a $(1-O(\epsilon))$-approximation algorithm in time
\[
 O(n(k+\epsilon^{-D}))+ (\epsilon^{-1} \log k)^{O(\epsilon^{-D})}\cdot k, 
\]
assuming that distance evaluations take constant time.
\end{theorem}
Notice that this is a linear-time efficient PTAS in the case where $k$ is considered constant. 
As $k$ is usually much smaller than $n$, this is a very desirable running time for applications with large $n$. 
We provide an improved algorithm for this particular problem for the sake of simplicity and because of its popularity, 
but point out that the additional techniques presented below are mostly standard and can be applied in varying degrees 
to the other two diversity problems and to a general constant $q\geq 1$. 

Recall that the algorithm presented in Section~\ref{s:PTAS1} is comprised of the following steps.
\begin{enumerate}
 \item ``Guess'' a constant-factor estimate $\Delta'$ of $\Delta$ (by exhaustive search).
 \item ``Guess'' the instance center $z_0$ (by exhaustive search).
 \item Define a ball $B'$ such that a) $\bar{B'}\subseteq \bar{B} \subseteq \opt$ and b) its radius is $O(\Delta)$. 
 \item Compute a cell decomposition of radius $O(\epsilon \Delta)$ over $B'$. 
 \item Perform exhaustive search, in the rounded instance, over the solutions that contain $\bar{B'}$.
\end{enumerate}

Notice that steps 3 to 5 are repeated for each combination of guesses made in the first two steps. 
In contrast, the new algorithm performs each step only once in a linear fashion. 
To achieve this, instead of finding the exact instance center $z_0$ we opt to approximate it by a point $z_0'$ that is close enough in distance. 
We perform a cell decomposition over the whole input \emph{before} searching for the approximate center $z_0'$ in order to decrease the search space. 
Finally, in the last step we reduce the search space again to find a close-to-optimal solution in the rounded instance.  

The new algorithm has the following steps.
\begin{enumerate}
 \item Compute a constant-factor estimate $\Delta'$ of $\Delta$.
 \item Compute a cell decomposition of radius $O(\epsilon \Delta)$ over the whole instance $X$. 
 \item Find a cell center $z_0'$ close to $z_0$, with $d(z_0,z_0')=O(\Delta)$. 
 \item Define a ball $B'$ such that a) $\bar{B'}\subseteq \bar{B} \subseteq  \opt$ and b) its radius is $O(\Delta)$.
 \item Perform a restricted search in the rounded instance, over the solutions that contain $\bar{B'}$.
\end{enumerate}



We start by stating a property of the main cluster $B$ which ensures that it is the unique highly-clustered region in the instance 
and hence it can easily be detected.

\begin{lemma}\label{lem:cluster} 
Given an instance of remote-clique, let $z_0$ and $B=B(z_0, 2\Delta)$ be the instance center and main cluster 
as defined in Theorem~\ref{thm:far=opt}. Then, $|\bar{B} | < k/2$. 

Moreover, for any ball $B(u,r)$ with $|X \setminus B(u,r)| < k/2$, we have $d(z_0,u)\leq 2\Delta+r$.
\end{lemma}

\begin{proof} 
Recall that point $z_0$ is chosen to be the center of the min-weight spanning star in $\opt_\cl$. 
By Theorem~\ref{thm:far=opt}, all points in $\bar{B}$ are in $\opt_\cl$, so if the first claim is false, 
we would have at least $k/2$ optimal points each at a distance $>2\Delta$ from $z_0$. Hence, 
 \[
  \st(\opt_\cl) = \sum_{u\in \opt_\cl} d(z_0,u) > \frac{k}{2} \cdot (2 \Delta) = k \cdot \Delta.
 \]
We now use inequality~\eqref{eq:relation1} to obtain
\[
 \cl(\opt_\cl) \geq \frac{k}{2} \st(\opt_\cl) > \frac{k^2}{2}\cdot \Delta > \binom{k}{2} \cdot \Delta = \cl(\opt_\cl),
\]
which is a contradiction. The second claim follows from the fact that balls $B(z_0, 2\Delta)$ and $B(u,r)$ must intersect, 
as each ball contains strictly more than half of the input points.
\end{proof}

We provide a detailed account of the new algorithm in the next two subsections: The first one describes steps 1 to 4 
and the second one describes step 5.

\subsection*{Approximating the main cluster}\label{ss:PTAS2cluster}

We fix a constant error parameter $\epsilon>0$.
The new algorithm finds an approximation $B'$ of the main cluster $B$ via the following steps:

Using the standard greedy algorithm for remote-clique over general metrics~\cite{birnbaum2009improved}, we obtain in time $O(nk)$ an estimate $\Delta'$ of $\Delta$ so that $\frac{1}{2}\Delta\leq \Delta'\leq \Delta$.

Then, for a constant $\delta$ with $1/\delta=\Theta(1/\epsilon)$, 
we compute a cell decomposition of radius $\delta \Delta'$ over the whole input $X$.  
Notice that by Lemma~\ref{lem:cluster}, this cell decomposition produces only $O(\frac{\Delta}{\delta \Delta'})^D = O(1/\epsilon)^D$ 
cell centers inside the main cluster $B$ and at most $k/2$ centers outside of it. 
Hence, it produces only $O(k+\epsilon^{-D})$ cells and it executes in time $O(n(k+\epsilon^{-D}))$. 
Let $\pi:X\rightarrow X$ be the function that maps each point to its cell center.

Next, we iterate over all cell centers in order to find a cell center $z_0'$ with the property that $|X\setminus B(z_0', 5\Delta')|< k/2$. 
Assuming that $\delta\leq 1$, Lemma~\ref{lem:cluster} guarantees the existence of at least one such cell center -- namely $\pi(z_0)$. 
This is because 
\[
5\Delta' \geq  \delta \Delta' + 4\Delta' \geq  d(z_0, \pi(z_0)) + 2\Delta, 
\]
so the ball $B(\pi(z_0), 5\Delta')$ contains $B=B(z_0, 2\Delta)$ which contains the required number of points.
This search takes time $O(n(k+\epsilon^{-D}))$. 

Once we have found such a cell center $z_0'$, we define the ball $B':=B(z_0', 13\Delta')$. 
By Lemma~\ref{lem:cluster}, we have $d(z_0, z_0')\leq 2\Delta + 5\Delta' \leq 9 \Delta'$. Therefore, 
\[
 13\Delta' \geq d(z_0, z_0') + 4\Delta' \geq d(z_0, z_0') + 2\Delta,
\]
which implies that $B'$ contains $B=B(z_0, 2\Delta)$, and so $\bar{B'} \subseteq \bar{B} \subseteq \opt$ by Theorem~\ref{thm:far=opt}.

We now invoke Lemma~\ref{lem:roundedset} to conclude that, since the radius of the cell decomposition is $\delta \Delta' \leq \delta \Delta$, 
for any $k$-set $T\subseteq X$, we have
\begin{align}\label{eq:roundedset2}
|\cl(T) - \cl(\hat{\pi}(T))| \leq 8\delta \cdot \cl(\opt) = O(\epsilon)\cdot \cl(\opt).\footnotemark
\end{align}
For the case $q=1$, Lemmas \ref{lem:mapping} and \ref{lem:roundedset} simplify and this bound can easily be improved to $2 \delta \cdot \cl(\opt)$. 

For simplicity, in what follows we assume that the images $\pi(B')$ and $\pi(\bar{B'})$ are disjoint. 
One possible way to achieve this is to enlarge set $B'$ so that it fully contains all cells that intersect with it.  
In any case, it is clear that the size of $\pi(B')$ is $O(\frac{13\Delta'}{\delta \Delta'})^D = O(\epsilon^{-D})$. 
The complexity of our algorithm, up to this point, is $O(n(k+\epsilon^{-D}))$.

\subsection*{Restricted search over the rounded instance} \label{ss:PTAS2search}

The last step of our algorithm is finding a $k$-set $T$ that contains $\bar{B'}$ and such that 
\begin{align}
  \cl(\hat{\pi}(T))\geq (1-\epsilon) \cdot \cl(\hat{\pi}(\opt)). \label{eq:approxOpt}
\end{align}
From inequalities~\eqref{eq:roundedset2} and \eqref{eq:approxOpt}, it easily follows that $\cl(T)\geq (1-O(\epsilon))\cdot \cl(\opt)$, as desired. 

For notational convenience, we represent a multiset over $\pi(X)$ by its vector of multiplicities $m\in\Z_{\geq 0}^{\pi(X)}$. 
For such a vector $m$, the remote-clique function is defined as 
\[
 \cl(m) = \frac{1}{2} \sum_{u,v\in \pi(X)} m_u m_v d(u,v).
\]
In the algorithm of the previous section, we consider all vectors $m\in\Z_{\geq 0}^{\pi(X)}$ of cardinality $\|m\|_1=k$ such that $m_u \leq \min\{|\pi^{-1}(u)|, k\}$ for each $u\in \pi(B')$ and $m_u = |\pi^{-1}(u)|$ for each $u\in \pi(\bar{B'})$. 
Once we find the vector $m$ of largest value $\cl(m)$, we find and return a corresponding pre-image, i.e. 
a $k$-set $T\subseteq X$ such that $m_u = |\pi^{-1}(u) \cap T| \ \forall u\in \pi(X)$. 

The new algorithm considers only a restricted collection of vectors $m\in\Z_{\geq 0}^{\pi(X)}$ defined as follows: 
For each point $u$ in $\pi(B')$, consider the list of multiplicities $(m_u^0, m_u^1, m_u^2, \cdots, 1, 0)$ where $m_u^0=\min\{|\pi^{-1}(u)|, k\}$ and $m^i_u = \min\{\lceil (1-\epsilon/2) m_u^{i-1}\rceil, m_u^{i-1} - 1\}$ for each $i\geq 1$. 
Then, for each vector $m$ with multiplicities coming from these lists in $\pi(B')$, with $m_u = |\pi^{-1}(u)|$ for each $u\in \bar{B'}$ 
and of cardinality $\|m\|_1 \leq k$, we raise its entries arbitrarily until $\|m\|_1=k$, but without exceeding any threshold $|\pi^{-1}(u)|$. Consider this final vector $m$. 

For each point $u$ in $\pi(B')$, the list of multiplicities we consider is of size $O(\epsilon^{-1} \log k)$ and hence the total number of generated vectors $m$ is $(\epsilon^{-1} \log k)^{O(\epsilon^{-D})}$. 
If $m^*$ is the vector of multiplicities of $\hat{\pi}(\opt)$, our search eventually considers a vector $m'$ such that
$m'_u\geq (1 - \epsilon/2) m^*_u$ for each $u\in \pi(X)$. 
The objective value of $m'$ is then
\begin{align*}
  \cl(m') &= \frac{1}{2} \sum_{u,v\in \pi(X)} m'_u m'_v d(u,v) \geq (1-\epsilon/2)^2 \cdot \frac{1}{2}  \sum_{u,v\in \pi(X)} m^*_u m^*_v d(u,v) \\
	    & \geq (1-\epsilon) \cdot \cl(m^*) =  (1-\epsilon) \cdot \cl(\hat{\pi}(\opt)).
 \end{align*}
This shows that our restricted search finds a $k$-set $T$ that observes~\eqref{eq:approxOpt}, as claimed. 

Finally, we argue that for each one of the $(\epsilon^{-1} \log k)^{O(\epsilon^{-D})}$ vectors $m \in\Z_{\geq 0}^{\pi(X)}$ that we generate, 
we can compute its objective value $\cl(m)$ in time linear in $k$. 
Indeed,
\[
 \cl(m) = \frac{1}{2} \sum_{u,v\in \pi(B')} m_u m_vd(u,v) + \sum_{u\in \pi(B'), v\in\pi(\bar{B'})} m_u m_vd(u,v) + \frac{1}{2} \sum_{u,v\in \pi(\bar{B'})} m_u m_v d(u,v).
\]
The first two terms on the right-hand side can be computed in time linear in $k$, 
while the last term (which is quadratic in $k$) is constant over all the generated vectors, so it needs to be computed only once. 
Hence, our algorithm takes a total time of  $O(n(k+ \epsilon^{-D})) + (\epsilon^{-1} \log k)^{O(\epsilon^{-D})} \cdot k$. 
This completes the description of the algorithm and the proof of Theorem~\ref{thm:PTAS2}.

\pagebreak[4]


\section{A PTAS for generalized min-bisection}\label{s:BP}

In this section we present a PTAS for the following problem: Given an even integer $k\geq 4$, a $k$-point metric space $(T,d)$ of constant doubling dimension $D$ and a constant $q\geq 1$, 
find a balanced bipartition $T=L\cup R$ that minimizes the expression $f(L,R):=\sum_{\ell\in L, r\in R} d^q(\ell, r)$. In this section, all mentioned bipartitions are assumed to balanced ($|L|=|R|=k/2$). As stated before, this and other results for bipartitions can be easily adapted to odd values of $k$. 
The minimum value of this expression is precisely $\bp^q(T)$.
When $q=1$, this problem is known as the metric min-bisection problem and it is NP-hard over general metrics~{\cite{fernandez2002polynomial}. 
Fernandez de la Vega et al.~\cite{fernandez2002polynomial} provided a PTAS for it for Euclidean distances 
of fixed dimension. Ours is a generalization of their technique for any constant $q\geq 1$ and any doubling metric space.

\begin{theorem}\label{thm:BP}
 For any constant $q\geq 1$, the $q$-th power version of the min-bisection problem over doubling metric spaces admits a PTAS.
\end{theorem}

We recall that this result, used as a black box to achieve an (approximate) oracle of the value of $\bp^q(T)$, 
is required by our main algorithm in Section~\ref{s:PTAS1}.

Fix a constant $q\geq 1$ and an error parameter $\epsilon>0$. 
Our algorithm will output a bipartition of value $(1+O(\epsilon))\cdot \bp^q(T).$
We perform grid rounding: We first apply a cell decomposition to round the input set $T$  
and then execute a restricted search over the solutions of the rounded instance. 
Interestingly, the following cell decomposition has two important differences compared to those used in previous sections: 
The cells have varying radii and the number of cells produced is super-constant in $k$.

\subsection*{Cell decomposition of varying radii}\label{ss:BPcell}

Let $\Delta:=\frac{4}{k^2} \cdot \bp^q(T) $ be the average value in the optimal solution. 
We start by computing a constant-factor estimation of $\Delta^{1/q}$. One way to do it is as follows: 
If $\Delta':= \frac{1}{2^q + 1}\cdot \frac{4}{k^2} \cdot \cl^q(T)$, then the two inequalities in~\eqref{eq:relation2} 
together with $k\geq 4$ ensure that $\frac{1}{2} \cdot \Delta^{1/q} \leq \Delta'^{1/q} \leq \Delta^{1/q}$.

Next, find the center $z$ of the minimum weight spanning star in $T$, so $\st^q(T)=\sum_{u\in T} d^q(z,u)$. 

We fix a constant $\delta$ with $1/\delta=\theta(2^q/\epsilon)$ and we perform the following cell decomposition of $T$. 
For a point $u\in T$ not yet assigned to a cell, define a new cell with $u$ as its center. Add to it any other point $v\in T$ not yet assigned to any other cell
and such that $d(v,u)\leq \delta\cdot \max\{\Delta'^{1/q}, \frac{1}{2} d(v,z)\}$. 
Repeat this operation until all points have an assigned cell and let $\pi:T\rightarrow T$ 
be the function that maps each point to its cell center.

\begin{lemma}
The above-defined cell decomposition produces $O((2^q / \epsilon)^D \log k)$ cells over $T$.
\end{lemma}

\begin{proof}
Recall that $1/\delta=\theta(2^q/\epsilon)$. We partition $T$ into the central ball $B(z,\Delta^{1/q})$ 
and the ``spherical shells'' $S_i:=B(z,2^i \cdot \Delta^{1/q}) \setminus B(z,2^{i-1} \cdot \Delta^{1/q})$, for $i\geq 1$. 
In the central ball, the distance between any two cell centers is strictly larger than $\delta \cdot \Delta'^{1/q}$, 
hence the number of cell centers in it is $O(\frac{\Delta^{1/q}}{\delta \cdot \Delta'^{1/q}})^D = O( 2^q / \epsilon )^D$. 
Similarly, for each $i\geq 1$, $S_i$ is contained in a ball of radius $2^i \cdot \Delta^{1/q}$ 
and the distance between cell centers in $S_i$ is strictly larger than  $2^{i-2} \delta \cdot \Delta^{1/q}$. 
Hence, the number of cell centers in $S_i$ is 
$O(\frac{2^i \cdot \Delta^{1/q}}{2^{i-2} \delta \cdot \Delta^{1/q}})^D = O( 2^q / \epsilon )^D$. 

Finally, we remark that $\Delta^{1/q}$ is within a factor of $k^2$ from the diameter of the instance, 
hence only the first $O(\log k)$ spherical shells are non-empty. This completes the proof. 
\end{proof}

What is the error incurred due to rounding? 
Notice that $d(u,\pi(u))\leq \max\{\Delta^{1/q}, \frac{1}{2} d(z,u)\}$ for each point $u\in T$.
The following are variations of Lemmas~\ref{lem:mapping} and \ref{lem:roundedset}.

\begin{lemma}\label{lem:roundingBP}
 Let $\pi:T\rightarrow T$ be such that $d(u,\pi(u))\leq \delta \cdot \max\{\Delta^{1/q}, \frac{1}{2} d(z,u)\}$ for each $u\in T$. 
 Then, for any points $u,v\in T$,
 \begin{align*}
  & |d^q(u,v) - d^q(\pi(u), \pi(v))| \leq 2^{q+1} \delta \cdot \Big[\Delta + d^q(u,v) + \frac{1}{2^q} d^q(z,u) + \frac{1}{2^q} d^q(z,v) \Big], 
 \end{align*}
and consequently, for any bipartition $T=L\cup R$, 
\begin{align}
 |f(L,R) - f(\hat{\pi}(L), \hat{\pi}(R))| \leq 2^{q+3}\delta \cdot f(L,R)=O(\epsilon) \cdot f(L,R). \label{eq:roundedset3}
\end{align}
\end{lemma}
\begin{proof}
 Consider first the case that $d(u,v)\leq d(\pi(u), \pi(v))$. Then,
 \begin{align*}
  d^q(\pi(u),\pi(v)) &\leq [d(\pi(u), u) + d(u,v) + d(v, \pi(v))]^q  \\
		    & \leq [d(u,v)+2\delta\cdot \max\{\Delta^{1/q}, \frac{1}{2} d(z,u), \frac{1}{2} d(z,v)\}]^q & \text{(by hyp.)} \\
		    & \leq d^q(u,v) + 2^{q+1}\delta\cdot \max\{ \Delta, d^q(u,v), \frac{1}{2^q} d^q(z,u), \frac{1}{2^q} d^q(z,v) \} & \text{(by~\eqref{eq:triangle3})} \\
		    & \leq d^q(u,v) + 2^{q+1}\delta \cdot \Big[\Delta + d^q(u,v) + \frac{1}{2^q} d^q(z,u) + \frac{1}{2^q} d^q(z,v) \Big],
 \end{align*}
which proves the first claim.  
Conversely, if $d(u,v)>d(\pi(u), \pi(v))$, a similar proof yields
\[
 d^q(u,v)\leq d^q(\pi(u),\pi(v)) + 2^{q+1}\delta \cdot [\Delta + d^q(\pi(u),\pi(v)) + \frac{1}{2^q} d^q(z,u) + \frac{1}{2^q} d^q(z,v)],
\]
which again gives the first claim after bounding the term $d^q(\pi(u),\pi(v))$ inside the brackets by $d^q(u,v)$.
Consider now a bipartition $T=L\cup R$. By the first claim and the definitions of $\Delta$ and $f(L,R)$, we get
\begin{align*}
 |f(L,R) - f(\hat{\pi}(L), \hat{\pi}(R))| &\leq \sum_{\ell \in L, r\in R} |d^q(\ell, r) - d^q(\pi(\ell), \pi(r))| \\
			       \leq & 2^{q+1}\delta \cdot \Big[ \bp^q(T) + f(L,R) + \frac{1}{2^q}\sum_{\ell\in L, r\in R} \big( d^q(z,l) + d^q(z,r) \big)  \Big].
\end{align*}
By our choice of point $z$ and by inequalities \eqref{eq:relation1} and \eqref{eq:relation2}, the last term is bounded by  
\begin{align*}
 \frac{1}{2^q}\sum_{\ell\in L, r\in R} \big( d^q(z,\ell) + d^q(z,r)\big) & = \frac{1}{2^q} \cdot \frac{k}{2} \sum_{u\in T} d^q(z,u) = \frac{1}{2^q} \cdot \frac{k}{2} \cdot \st^q(T) \\
									 & \leq \frac{1}{2^q} \cdot \cl^q(T) \leq \frac{2^q+1}{2^q} \cdot \bp^q(T) < 2 \cdot \bp^q(T).
\end{align*}
This bound, together with the bound $\bp^q(T)\leq f(L,R)$, yields the second claim.
\end{proof}

\subsection*{Restricted search over the rounded instance}

The second step of the algorithm aims to find a bipartition $T=L\cup R$ so that 
\begin{align}
 f(\hat{\pi}(L), \hat{\pi}(R))\leq (1+O(\epsilon)) \cdot \bp^q(\hat{\pi}(T)), \label{eq:approxOpt2}
\end{align}
by performing a restricted search over the space of bipartitions of the rounded instance $\hat{\pi}(T)$.
This inequality, together with~\eqref{eq:roundedset3} implies that $f(L,R)\leq (1+O(\epsilon)) \cdot \bp^q(T)$, as desired.

As we did before, we represent a multiset over $\pi(T)$ by its vector of multiplicities $m\in\Z_{\geq 0}^{\pi(T)}$. 
Let $M$ be the vector of multiplicities of the full multiset $\hat{\pi}(T)$, i.e. $M_u = |\pi^{-1}(u)|$ for each $u\in \pi(T)$. 
For any vector $m$ with $0\leq m\leq M$, where the inequalities hold component-wise, we define its complement as $\bar{m} := M - m$. 
A bipartition is then given by a vector $m$ with $0\leq m\leq M$ and of cardinality $\|m\|_1 = k/2$. Its objective value is
\[ 
  f(m,\bar{m}) = \sum_{u,v\in \pi(T)} m_u \bar{m}_v d^q(u,v). 
\]

We define a restricted collection of bipartitions over which we search as follows: 
We fix a constant $\delta'>0$ with $1/\delta'=\theta(2^q/\epsilon)$. 
Then, for each point $u\in \pi(T)$ we generate a list of candidate multiplicities inside the range $[0, M_u]$ at regular intervals of size 
$\max\big\{\big \lfloor \delta' \cdot M_u \big\rfloor, 1\big\}$. 
For each vector $m$ with entries coming from these lists and of cardinality $\|m\|_1\leq k/2$, 
we raise its entries in an arbitrary way until $\|m\|_1=k/2$, 
but without increasing any entry $m_u$ by more than $\delta' \cdot M_u$ nor exceeding the threshold $M_u$. Consider this final vector $m$. If this last operation of raising entries is not possible, we discard vector $m$.
Among the vectors thus generated, we select the vector $m$ with smallest objective value and output a corresponding pre-image, 
i.e.~a bipartition $T=L\cup R$ where the vector of multiplicities of $\hat{\pi}(L)$ is precisely $m$.

Notice that the list of candidate multiplicities for each $u\in \pi(T)$ is of size $O(1/\delta')=O(2^q/\epsilon)$, 
hence the total number of bipartitions we generate is $(2^q/\epsilon)^{O((2^q / \epsilon)^D \log k )} = k^{O((2^q / \epsilon)^D \log (2^q / \epsilon))}$ which is polynomially bounded.
Moreover, if $m^*$ represents the optimal bipartition in the rounded instance, i.e. $f(m^*, \bar{m}^*)=\bp^q(\hat{\pi}(T))$, 
then we are guaranteed to generate a vector $m'$ such that 
\[
 m^*_u - \delta'\cdot M_u \leq m' \leq m^*_u + \delta'\cdot M_u \quad \text{ and } \quad \bar{m}^*_u - \delta'\cdot M_u \leq \bar{m}'_u \leq \bar{m}^*_u + \delta'\cdot M_u
\]
for each point $u\in \pi(T)$. The objective value of this vector is then 
\begin{align*}
 f(m',\bar{m}') &= \sum_{u,v\in \pi(T)} m'_u \bar{m}'_v d^q(u,v) \leq \sum_{u,v\in \pi(T)} (m^*_u + \delta' \cdot M_u) \bar{m}'_v d^q(u,v) \\
	      & \leq \sum_{u,v\in \pi(T)} m^*_u \bar{m}'_v d^q(u,v) + \delta'\sum_{u,v\in \pi(T)} M_u M_v d^q(u,v) \\
	      & \leq \sum_{u,v\in \pi(T)} m^*_u (\bar{m}^*_v + \delta'\cdot M_v) d^q(u,v) + \delta'\sum_{u,v\in \pi(T)} M_u M_v d^q(u,v) \\
	      & \leq \sum_{u,v\in \pi(T)} m^*_u \bar{m}^*_v d^q(u,v) + 2\delta'\sum_{u,v\in \pi(T)} M_u M_v d^q(u,v) \\
	      & = \bp^q(\hat{\pi}(T)) + 4\delta' \cdot \cl^q(\hat{\pi}(T)) \\
	      & \leq \bp^q(\hat{\pi}(T)) + 4(2^q+1)\delta' \cdot \bp^q(\hat{\pi}(T)) = (1+O(\epsilon))\cdot \bp^q(\hat{\pi}(T)),
\end{align*}
where the last inequality is an application of~\eqref{eq:relation2} over the rounded instance $\hat{\pi}(T)$. 
This shows that our output bipartition observes~\eqref{eq:approxOpt2} as claimed which completes the proof of Theorem~\ref{thm:BP}.

\section{Hardness Results}\label{s:NP}

In this section we present two hardness results for our diversity maximization problems. 

We first show that none of the three problems nor their generalizations for $q\geq 1$ admits a PTAS over general metrics, 
under the assumption that the planted-clique problem is hard~\cite{alon2011inapproximability}. 
Thus, the assumption of fixed doubling dimension is necessary to achieve our strong results.  
This is a generalization of a similar hardness result for standard remote-clique in~\cite{borodin2012max}.

\begin{theorem}\label{thm:planted}
Under the assumption that the planted clique problem is hard and for any fixed $q\in \R_{\geq 1}$, 
neither remote-clique, remote-star nor remote-bipartition over general metrics 
admits an approximation algorithm with a constant factor higher than $2^{-q}$.
\end{theorem}

\begin{proof}
Based on the assumption that the planted clique problem is hard, it is proven by Alon et al.~\cite{alon2011inapproximability} that for any constants $0 < \epsilon \leq \frac{2}{3}$ and $\delta>0$, it is hard to distinguish between 
a) an $n$-vertex graph containing a clique of size $k = n^{1-\epsilon}$, and 
b) an $n$-vertex graph where each $k$-vertex subset induces at most $\delta \cdot \binom{k}{2}$ edges. 

Fix a constant $q\geq 1$ and one of the three problems, and assume by contradiction that there is a constant $\lambda>0$ 
and an approximation algorithm for this problem with a factor of $2^{-q}+\lambda$. 
We show how we can use this algorithm to distinguish between the cases a) and b) mentioned above, 
where we set $\delta := \lambda$ (and $\epsilon$ can be set to any value in its prescribed range). 
For a given $n$-vertex graph $G=(V,E)$, we define its associated $(1,2)$-metric space $(V,d)$ 
by setting $d(u,v)=2$ if $u$ and $v$ are adjacent, and 1 otherwise. 

In case a), the space $(V,d)$ contains a $k$-set $T$ with all pair-wise distances equal to 2. 
Therefore, $\Delta_{\cl^q} = \Delta_{\st^q} = \Delta_{\bp^q} = 2^q$. Our assumed approximation algorithm is guaranteed to find a solution with an average value of at least 
$(2^{-q} + \lambda)\cdot 2^q = 1+ 2^q\cdot \delta$. 

In case b), for any $k$-set $T$ in the space $(V,d)$ we have that 
$\cl^q(T)\leq (1\cdot (1-\delta) + 2^q \cdot \delta ) \cdot \binom{k}{2} < (1+ 2^q\cdot \delta) \cdot \binom{k}{2}$. 
By inequalities \eqref{eq:relation1} and \eqref{eq:relation2}, this implies that 
$\Delta_{\st^q} \leq \Delta_{\cl^q} < 1+ 2^q\cdot \delta$ and 
$\Delta_{\bp^q} < \Delta_{\cl^q} < 1+2^q\cdot\delta$.
In particular, in this case our assumed algorithm is bound to return a solution with an average value strictly smaller 
than $1+2^q\cdot\delta$. Thus we can distinguish this case from the one above. This completes the proof.
\end{proof}

Our next result corresponds to the first proof of NP-hardness for any of the three diversity problems in a fixed-dimensional setting. In fact, the only other diversity maximization problem known to be NP-hard in such a setting is remote-edge~\cite{wang1988study}. 
In particular, we prove NP-hardness for the squared distances ($q=2$) version of remote-clique 
in the case where all input points are \emph{unit vectors} in the Euclidean space $\R^3$, i.e. $X\subseteq \s^2$.

\begin{theorem}\label{thm:NP}
 The squared distances version ($q=2$) of the remote-clique problem is NP-hard over the three-dimensional Euclidean space.
\end{theorem}

We remark that squared Euclidean distances over unit vectors correspond precisely to the popular cosine distances, 
hence the case considered is highly relevant.

For a $k$-set $T\subseteq \s^2$ with Euclidean distances, the function $\cl^2(T):=\sum_{\{u,v\} \in \binom{T}{2}} d^2(u,v)$ 
has very particular geometric properties related to the concept of \emph{centroid}. 
The centroid of a $k$-set $T$ is defined as $z_T:=\frac{1}{k}\sum_{u\in T} u$. 
It represents the coordinate-wise average of the points in $T$. 
The following result greatly simplifies the computation of function $\cl^2(T)$ in terms of the centroid. 
We state it for a general dimension $D$ even if we only use it for the case $D=3$. 
Its proof is deferred to Appendix~\ref{s:proofs}.

\begin{lemma} \label{lem:centroid}
For a $k$-set $T\subseteq \s^{D-1} \subseteq\mathbb{R}^D$ with centroid $z_T:=\frac{1}{k}\sum_{u\in T} u$,
\begin{align*}
\cl^2(T) =k^2 \cdot \big(1- \|z_T\|^2\big).
\end{align*}
\end{lemma}

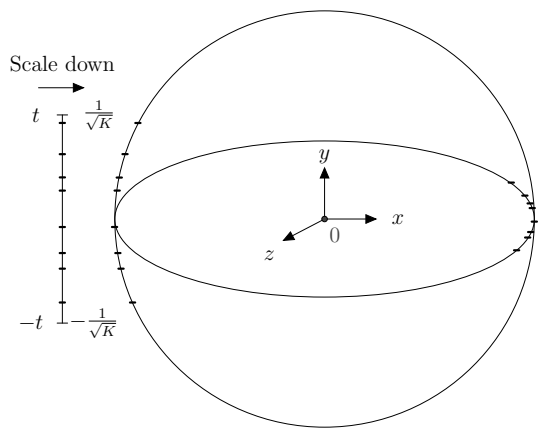
\begin{figure}[t]
\begin{center}
\scalebox{.69}{

\definecolor{uuuuuu}{rgb}{0.26666666666666666,0.26666666666666666,0.26666666666666666}
\begin{tikzpicture}[line cap=round,line join=round,>=triangle 45,x=1.0cm,y=1.0cm]
\clip(-6.46166865124929,-4.56696449800584) rectangle (4.584858441716284,4.459975924223601);
\draw(0.,0.) circle (4.cm);
\draw(0.,0.) ellipse (4.cm and 1.5cm);
\draw [->] (0.,0.) -- (1.,0.);
\draw [->] (0.,0.) -- (0.,1.);
\draw [->] (0.,0.) -- (-0.8,-0.42);
\draw (-5.,2.)-- (-5.,-2.);
\draw (-5.1,2.)-- (-4.9,2.);
\draw (-5.1,-2.)-- (-4.9,-2.);
\draw [line width=1.2pt] (-5.05,0.8)-- (-4.95,0.8);
\draw [line width=1.2pt] (-5.050810991685095,0.5460367734292512)-- (-4.9508109916850955,0.5460367734292512);
\draw [line width=1.2pt] (-5.050810991685095,-0.9539632265707487)-- (-4.9508109916850955,-0.9539632265707487);
\draw [line width=1.2pt] (-5.050810991685099,-1.6039632265707493)-- (-4.950810991685099,-1.6039632265707493);
\draw [line width=1.2pt] (-5.050810991685095,1.8460367734292518)-- (-4.9508109916850955,1.8460367734292518);
\draw [line width=1.2pt] (-5.050810991685097,-0.15396322657074896)-- (-4.950810991685097,-0.15396322657074896);
\draw [line width=1.2pt] (-5.050810991685095,-0.6539632265707488)-- (-4.9508109916850955,-0.6539632265707488);
\draw [line width=1.2pt] (-5.050810991685095,1.2460367734292512)-- (-4.9508109916850955,1.2460367734292512);
\draw [line width=1.2pt] (-3.6094445692663104,1.8451393713382171)-- (-3.5094445692663108,1.8451393713382171);
\draw [line width=1.2pt] (-3.855012190720009,1.2470700845582339)-- (-3.755012190720009,1.2470700845582339);
\draw [line width=1.2pt] (-3.965080145238988,0.7966690187494272)-- (-3.865080145238988,0.7966690187494272);
\draw [line width=1.2pt] (-4.005080145238989,0.5466690187494273)-- (-3.905080145238989,0.5466690187494273);
\draw [line width=1.2pt] (-4.05508014523899,-0.15333098125057268)-- (-3.9550801452389903,-0.15333098125057268);
\draw [line width=1.2pt] (-3.9913479665318663,-0.6542010673871617)-- (-3.891347966531866,-0.6542010673871617);
\draw [line width=1.2pt] (-3.9313479665318654,-0.9542010673871616)-- (-3.8313479665318653,-0.9542010673871616);
\draw [line width=1.2pt] (-3.7113479665318634,-1.604201067387162)-- (-3.6113479665318633,-1.604201067387162);
\draw [line width=1.2pt] (3.6117997767561913,-0.6007537646574633)-- (3.7117997767561923,-0.6007537646574633);
\draw [line width=1.2pt] (3.8228999299538233,-0.3556221917608715)-- (3.9228999299538234,-0.3556221917608715);
\draw [line width=1.2pt] (3.88,-0.25)-- (3.98,-0.25);
\draw [line width=1.2pt] (3.95,-0.05)-- (4.05,-0.05);
\draw [line width=1.2pt] (3.91,0.21)-- (4.01,0.21);
\draw [line width=1.2pt] (3.87,0.3)-- (3.97,0.3);
\draw [line width=1.2pt] (3.77,0.45)-- (3.87,0.45);
\draw [line width=1.2pt] (3.51,0.7)-- (3.61,0.7);
\draw [->,line width=0.4pt] (-5.4549694172149525,2.52233311158679) -- (-4.578458538348341,2.52233311158679);
\begin{scriptsize}
\draw [fill=uuuuuu] (0.,0.) circle (1.5pt);
\draw[color=uuuuuu] (0.2,-0.3) node {\large $0$};
\draw[color=black] (1.4,0) node {\large $x$};
\draw[color=black] (0,1.2) node {\large $y$};
\draw[color=black] (-1.05,-0.7) node {\large $z$};
\draw[color=black] (-5.5,2) node {\large $t$};
\draw[color=black] (-5.6,-2) node {\large $-t$};
\draw[color=black] (-4.4,-2) node {\large $-\frac{1}{\sqrt{K}}$};
\draw[color=black] (-4.3,2) node {\large $\frac{1}{ \sqrt{K}}$};
\draw[color=black] (-5,3) node {\large Scale down};
\end{scriptsize}
\end{tikzpicture}

}
\end{center}
\caption{Reduction from $K$-SUM to remote-clique with $q=2$, $|X|=2|M|$ and $k=2K$.}
\label{fig:NP}
\end{figure}

We present a reduction from the $K$-SUM problem which is known to be NP-hard: Given a set $M$ of integer numbers in the range $[-t, t]$ for some threshold $t$ and a positive integer $K$, determine whether there is a $K$-set $S\subseteq M$ that sums to zero.
Given such an instance of $K$-SUM, we define the following instance $X\subseteq \s^2$ of remote-clique  
with $q=2$, $|X|=2|M|$ and $k=2K$, see Figure \ref{fig:NP}. 
For each $m\in M$, set $m':=\frac{m}{t\sqrt{K}}$ and define
\[
 X := \Big\{\ell_m:=\big(-\sqrt{1 - m'^2}, m', 0\big)^\intercal : \ m\in M\Big\} \cup \Big\{r_m:=\big(\sqrt{1-m'^2}, 0, m'\big)^\intercal: \ m\in M\Big\}.
\]
Due to the scaling down by a factor of $\frac{1}{t\sqrt{K}}$, the $y$- and $z$-components of all points in $X$ are upper bounded by $\frac{1}{\sqrt{K}}$ in absolute value, 
while their $x$-components are lower bounded by $\sqrt{1-\frac{1}{K}}$ in absolute value. 
The points are thus tightly clustered around one of the two antipodal points $\pm(1,0,0)$ and $X$ is partitioned into a \emph{left cluster} and a \emph{right cluster}.

From Lemma~\ref{lem:centroid}, it is clear that solving this instance of remote-clique is equivalent 
to finding the $k$-set whose centroid is closest to the origin. 
Hence, the proof of Theorem~\ref{thm:NP} is complete once we show the following claim.

\begin{lemma}
 If $M$ has a $K$-set $S$ with zero sum, then $X$ has a $k$-set $T$ with centroid $z_T=0$. 
 Otherwise, for every $k$-set $T\subseteq X$ we have $\|z_T\|\geq \frac{1}{2tK^{3/2}}$.
\end{lemma}
\begin{proof}
 Suppose that $M$ has a $K$-set $S$ with zero sum and define the $k$-set $T:=\{\ell_m, r_m: \ m\in S \}\subseteq X$. 
Recall that its centroid $z_T$ corresponds to the component-wise average of the points in $T$, so we analyze these components separately. 
In $z$, all points of $T$ on the left cluster are zero and those on the right cluster have a zero sum, so $(z_T)_z=0$. 
In $y$, all points of $T$ on the right cluster are zero and those on the left cluster have a zero sum, so $(z_T)_y=0$. 
And in $x$, each point $\ell_m$ of $T$ on the left cluster is canceled out by its paired point $r_m$ on the right cluster, so $(z_T)_x=0$. Therefore, $z_T=0$.

Finally, we prove the contrapositive of the second statement, i.e. we assume that there is a $k$-set $T\subseteq X$ with $\|z_T\|<\frac{1}{2tK^{3/2}}$. 
The set $T$ must contain exactly $K$ points in the left cluster and $K$ points in the right cluster. 
Indeed, if $T$ had at most $K-1$ points in the left cluster, then the $x$-component of its centroid would give
\[
(z_T)_x \geq (K-1)(-1)+(K+1)\sqrt{1-\frac{1}{K}} \geq -(K-1)+(K+1)\left(1-\frac{1}{K}\right)=1-\frac{1}{K},
\] 
and hence $\|z_T\| \geq |(z_T)_x | \geq 1 - \frac{1}{K} > \frac{1}{2t K^{3/2}}$ for $K\geq 2$ and $t\geq 1$, leading to a contradiction. 

Let $T= L\cup R$ be the corresponding (balanced) bipartition of $T$ given by the left and right clusters. 
Each of $L$ and $R$ must correspond to a $K$-set of $M$ with zero sum. 
Otherwise, without loss of generality $L$ corresponds to a $K$-set $S$ of $M$ with sum at least 1, but then 
\[
(z_T)_y=\frac{1}{2K}\sum_{m\in S} m' = \frac{1}{2tK^{3/2}} \sum_{m\in S} m \geq \frac{1}{2tK^{3/2}}
\]
and thus $\|z_T\|\geq |(z_T)_y| \geq \frac{1}{2t K^{3/2}}$, again a contradiction. This completes the proof.
\end{proof}

\section{Acknowledgements} \label{s:Ack}
Friedrich Eisenbrand acknowledges support from the Swiss National Science Foundation grant 163071, ``Convexity, geometry of numbers, and the complexity of integer programming''. This work was done while Sarah Morell was affiliated to EPFL and was visiting the Simons Institute for the Theory of Computing.


\appendix
\section{Deferred proofs}\label{s:proofs}

\begin{proof}[Proof of Lemma~\ref{lem:relaxedtriangle}]
 We can assume without loss of generality that $d(u,v)\geq d(v,w)$. 
 If $t:=d(v,w)/d(u,v)$, then $t$ is a variable between 0 and 1. The claimed inequality~\eqref{eq:triangle1} then reduces to 
$(1+t)^q \leq 2^{q-1}[1+t^q]$ which in turn is equivalent to $f(t):=\frac{(1+t)^q}{1+t^q} \leq 2^{q-1}$. 
The inequality clearly holds for $t=1$. As $f(t)$ is monotone increasing in the interval $[0,1]$, it must hold in the full interval. 
The monotonicity of $f(t)$ is verified by checking that its first derivative is non-negative.

We now pass to the proof of inequality~\eqref{eq:triangle3}. We consider two cases. If $x\leq \epsilon y$, then 
\[
 (x+\epsilon y)^q \leq (2\epsilon y)^q \leq 2^q \epsilon \cdot y^q \leq x^q + 2^q\epsilon \cdot \max\{ x^q, y^q \}.
\]
If $x>\epsilon y$, we use the binomial series of the term $(1+\epsilon y/x)^q$ in order to obtain
\begin{align*}
 (x+\epsilon y)^q &= x^q (1+ \epsilon y/x)^q = x^q \sum_{j=0}^\infty \binom{q}{j} (\epsilon y /x)^j \\
		  &\leq x^q \Big[1+\epsilon \sum_{j=1}^\infty \binom{q}{j} (y/x)^j \Big] = x^q + \epsilon \sum_{j=1}^\infty \binom{q}{j} x^{q-j}y^j \\
		  & \leq x^q + \epsilon \cdot \max\{x^q,y^q\}\sum_{j=1}^\infty \binom{q}{j} < x^q + 2^q\epsilon \cdot \max\{x^q, y^q\}.
\end{align*}
This completes the proof of inequality~\eqref{eq:triangle3}.
\end{proof}

\begin{proof}[Proof of Lemma~\ref{lem:objectiverelations}]
By the minimality of the objective function for remote-star, 
we have $\st^q (T)\leq \sum_{v\in T} d^q(u,v)$ for each point $u\in T$. Therefore,
\[
 \cl^q(T)=\frac{1}{2}\sum_{u\in T} \sum_{v\in T} d^q(u,v) \geq \frac{1}{2} \sum_{u\in T} \st^q(T) = \frac{k}{2}\cdot \st^q(T),
\]
which proves the first inequality in~\eqref{eq:relation1}. 
For the second inequality, let $z$ be the center of the min-weight spanning star in $T$ so that $\st^q(T)=\sum_{v\in T} d^q(z,v)$. 
We apply~\eqref{eq:triangle1} to obtain
\[
 \cl^q(T) = \frac{1}{2}\sum_{u,v\in T} d^q(u,v) \leq \frac{1}{2} \sum_{u,v\in T} 2^{q-1}[d^q(z,u)+d^q(z,v)] = 2^{q-1} k \cdot \st^q(T). 
\]

Next, the first inequality in~\eqref{eq:relation2} follows from the observation that, 
if we generate a bipartition of set $T$ uniformly at random, 
the probability of any pair of points to appear on opposite sides of the bipartition is $(k^2/4) / \binom{k}{2} = \frac{k}{2(k-1)}$. 
By an averaging argument, there must be a bipartition whose weight is at most $\frac{k}{2(k-1)} \cdot \cl^q(T)$. 
It remains to prove the second inequality in~\eqref{eq:relation2}. 
Let $T=L\cup R$ be the min-weight bipartition of $T$, so $\bp^q(T)=\sum_{\ell\in L, r\in R} d^q(\ell,r)$. 
An application of inequality~\eqref{eq:triangle1} gives $d^q(\ell, \ell') \leq 2^{q-1}[d^q(\ell, r)+ d^q (\ell', r)]$. 
We sum up these inequalities over all $\ell, \ell'\in L$ and $r\in R$ to obtain
 \[
 |R| \cdot \sum_{\ell, \ell' \in L} d^q(\ell, \ell') \leq 2^{q-1} \cdot 2\cdot |L| \cdot \sum_{\ell\in L, r\in R} d^q(\ell, r).
 \]
Since $|L|=|R|$, the previous inequality reduces to $\cl^q(L) \leq 2^{q-1} \cdot \bp^q(T)$. 
Similarly, one can prove that $\cl^q(R) \leq 2^{q-1} \cdot \bp^q(T)$. 
Therefore, if we classify all $\binom{|T|}{2}$ pairs appearing in $\cl^q(T)$ into three groups 
(those inside $L$, those inside $R$ and those in $L\times R$), we obtain
\[
 \cl^q(T)=\cl^q(L)+\cl^q(R)+\bp^q(T) \leq (2^q + 1)\cdot \bp^q(T).
\]
This completes the proof of \eqref{eq:relation2}.
\end{proof}

\begin{proof}[Proof of Lemma~\ref{lem:roundedset}]
 The proof for remote-clique is straightforward:
 \begin{align*}
  |\cl^q(T) - \cl^q(\hat{\pi}(T))| &\leq \sum_{\{u,v\}\in \binom{T}{2} } |d^q(u,v) - d^q(\pi(u), \pi(v))| \\
			    & \leq 2^{q+1} \delta \bigg[\binom{k}{2}\Delta + \sum_{\{u,v\}\in \binom{T}{2} } d^q (u,v) \bigg] & \text{(by Lemma~\ref{lem:mapping})} \\
			    & = 2^{q+1} \delta \big[ \cl^q(\opt)+\cl^q(T) \big]. & \text{(by def. of $\Delta=\Delta_{\cl^q}$)} 
 \end{align*}

 We pass to remote-star.  If $\st^q(T)\leq \st^q(\hat{\pi}(T))$, let $z$ be the center of the min-weight spanning star in $T$ so that $\st^q(T)=\sum_{u\in T\setminus \{z\}} d^q(z,u)$ 
 and $\st^q(\hat{\pi}(T))\leq \sum_{u\in T\setminus \{z\}} d^q(\pi(z),\pi(u))$. Hence, 
 \begin{align*}
  \st^q(\hat{\pi}(T)) - \st^q(T) & \leq \sum_{u\in T\setminus\{z\}} \big[ d^q(\pi(z), \pi(u)) - d^q(z,u) \big] & \text{(by our choice of $z$)}\\
			    & \leq 2^{q+1} \delta \bigg[ (k-1)\Delta + \sum_{u\in T \setminus \{z\}} d^q(z,u)  \bigg] & \text{(by Lemma~\ref{lem:mapping})} \\
			    &= 2^{q+1} \delta \big[ \st^q(\opt) + \st^q(T) \big]. & \text{(by def. of $\Delta=\Delta_{\st^q}$)}
 \end{align*}
 If $\st^q(\hat{\pi}(T)) < \st^q(T)$, let $z\in T$ be a point so that $\pi(z)$ is the center of the min-weight spanning in $\hat{\pi}(T)$. 
 Then, $\st^q(\hat{\pi}(T))= \sum_{u\in T\setminus \{z\}} d^q(\pi(z),\pi(u))$ and $\st^q(T)\leq \sum_{u\in T\setminus \{z\}} d^q(z,u)$ and 
 \begin{align*}
  \st^q(T) - \st^q(\hat{\pi}(T))  & \leq \sum_{u\in T\setminus\{z\}} \big[ d^q(z,u) - d^q(\pi(z), \pi(u)) \big] & \text{(by our choice of $z$)}\\
			    & \leq 2^{q+1} \delta \bigg[ (k-1)\Delta + \sum_{u\in T \setminus \{z\}} d^q(\pi(z),\pi(u))  \bigg] & \text{(by Lemma~\ref{lem:mapping})} \\
			    &= 2^{q+1} \delta \big[ \st^q(\opt) + \st^q(\hat{\pi}(T)) \big] & \text{(by def. of $\Delta=\Delta_{\st^q}$)} \\
			    &\leq 2^{q+1} \delta \big[ \st^q(\opt) + \st^q(T) \big].
 \end{align*}

We finally pass to remote-bipartition. If $\bp^q(T)\leq \bp^q(\hat{\pi}(T))$, let $T=L\cup R$ be the min-weight bipartition of $T$ so that 
$\bp^q(T)=\sum_{\ell\in L, r\in R} d^q(\ell,r)$, and $\bp^q(\hat{\pi}(T))\leq \sum_{\ell\in L, r\in R} d^q(\pi(\ell), \pi(r))$. Then,
\begin{align*}
 \bp^q(\hat{\pi}(T)) - \bp^q(T)   & \leq \sum_{\ell\in L, r\in R} \big[ d^q(\pi(\ell), \pi(r)) - d^q(\ell, r) \big] & \text{(by our choice of $L,R$)}\\
			    & \leq 2^{q+1} \delta \bigg[ (k^2/4)\Delta+\sum_{\ell\in L, r\in R} d^q(\ell,r) \bigg] & \text{(by Lemma~\ref{lem:mapping})} \\
			    &=2^{q+1} \delta \big[ \st^q(\opt)+ \bp^q(T) \big]. & \text{(by def. of $\Delta=\Delta_{\bp^q}$)}
\end{align*}
The proof for the case $\bp^q(\hat{\pi}(T))< \bp^q(T)$ is similar.
\end{proof}

\begin{proof} [Proof of Lemma~\ref{lem:centroid}]
For a fixed $k$-set $T\subseteq \s^{D-1}$, let $f_T(x):=\sum_{u\in T} d^2(x,u)$ for $x\in \R^D$.
\begin{align}
 f_T(x) &= \sum_{u\in T} \|x-u\|^2 = \sum_{u\in T} \|(x-z_T) + (z_T - u)\|^2 \nonumber \\
	&= \sum_{u\in T} \left[ \|x-z_T\|^2 + \|z_T - u\|^2 + 2(x-z_T)^\intercal (z_T - u) \right] \nonumber \\
	& = k\cdot d^2(x, z_T) + f_T(z_T) + 2(x-z_T)^\intercal \Big(k\cdot z_T - \sum_{u\in T} u \Big) \nonumber \\
	& = k\cdot d^2(x, z_T) + f_T(z_T), \label{eq:centroid1} 
\end{align}
where the last term on the right-hand side vanishes, because the expression in parenthesis is zero by definition of centroid. 
Next, we use the identity $\cl^2(T) = \frac{1}{2} \sum_{u\in T} \sum_{v\in T} d^2(u,v) = \frac{1}{2}\sum_{u\in T} f_T(u)$ and \eqref{eq:centroid1}. We obtain 
\begin{align}
 \cl^2(T) = \frac{1}{2} \sum_{u\in T} \left[ k\cdot d^2(u, z_T) + f_T(z_T) \right] = \frac{1}{2} \left[ k\cdot f_T(z_T) + k \cdot f_T(z_T) \right] = k \cdot f_T(z_T). \label{eq:centroid2}
\end{align}
Finally, if we take $x=0$, then $f_T(0)=k$ because $T\subseteq \s^2$. From~\eqref{eq:centroid1}, we obtain 
\[
 f_T(z_T) = f_T(0) - k\cdot d^2(0,z_T) = k\cdot \big( 1 - \|z_T\|^2  \big),
\]
and hence~\eqref{eq:centroid2} gives $\cl^2(T) = k^2\cdot \big( 1 - \|z_T\|^2  \big)$, as claimed.
\end{proof}

\bibliographystyle{plain}
\small
\bibliography{lit}

\begin{thebibliography}{10}

\bibitem{abbassi2013diversity}
Z.~Abbassi, V.~S. Mirrokni, and M.~Thakur.
\newblock Diversity maximization under matroid constraints.
\newblock In {\em 19th Conference on Knowledge Discovery and Data Mining
  (SIGKDD)}, pages 32--40. ACM, 2013.

\bibitem{aghamolaei2015diversity}
S.~Aghamolaei, M.~Farhadi, and H.~Zarrabi-Zadeh.
\newblock Diversity maximization via composable coresets.
\newblock In {\em 27th Canadian Conference on Computational Geometry (CCCG)},
  page~43, 2015.

\bibitem{alon2011inapproximability}
N.~Alon, S.~Arora, R.~Manokaran, D.~Moshkovitz, and O.~Weinstein.
\newblock Inapproximability of densest $\kappa$-subgraph from average case
  hardness.
\newblock {\em Unpublished manuscript}, 2011.

\bibitem{Bartal_2016_dimension}
A.~Bartal and L.~A. Gottlieb.
\newblock Dimension reduction techniques for $\ell_p$, $1\leq p\leq 2$, with
  applications.
\newblock In {\em Proceedings of the 32nd Symposium on Computational Geometry
  (SoCG)}, pages 16:1--16:15, 2016.

\bibitem{Bartal:2011:DRB:2133036.2133104}
Y.~Bartal, B.~Recht, and L.~J. Schulman.
\newblock Dimensionality reduction: Beyond the {Johnson-Lindenstrauss} bound.
\newblock In {\em Proceedings of the Twenty-second Annual ACM-SIAM Symposium on
  Discrete Algorithms (SODA)}, pages 868--887, 2011.

\bibitem{bhaskara2016linear}
A.~Bhaskara, M.~Ghadiri, V.~Mirrokni, and O.~Svensson.
\newblock Linear relaxations for finding diverse elements in metric spaces.
\newblock In {\em Advances in Neural Information Processing Systems}, pages
  4098--4106, 2016.

\bibitem{birnbaum2009improved}
B.~Birnbaum and K.~J. Goldman.
\newblock An improved analysis for a greedy remote-clique algorithm using
  factor-revealing {LP}s.
\newblock {\em Algorithmica}, 55(1):42--59, 2009.

\bibitem{borodin2012max}
A.~Borodin, H.~C. Lee, and Y.~Ye.
\newblock Max-sum diversification, monotone submodular functions and dynamic
  updates.
\newblock In {\em Proceedings of the 31st Symposium on Principles of Database
  Systems}, pages 155--166, 2012.

\bibitem{ceccarello2017mapreduce}
M.~Ceccarello, A.~Pietracaprina, G.~Pucci, and E.~Upfal.
\newblock Mapreduce and streaming algorithms for diversity maximization in
  metric spaces of bounded doubling dimension.
\newblock {\em Proceedings of the VLDB Endowment}, 10(5):469--480, 2017.

\bibitem{cevallos_2016_max-sum}
A.~Cevallos, F.~Eisenbrand, and R.~Zenklusen.
\newblock Max-sum diversity via convex programming.
\newblock In {\em 32nd Annual Symposium on Computational Geometry (SoCG)},
  pages 26:1--26:14, 2016.

\bibitem{cevallos2017local}
A.~Cevallos, F.~Eisenbrand, and R.~Zenklusen.
\newblock Local search for max-sum diversification.
\newblock In {\em 28th Symposium on Discrete Algorithms (SODA)}, pages
  130--142. SIAM, 2017.

\bibitem{chandra2001approximation}
B.~Chandra and M.~M. Halld{\'o}rsson.
\newblock Approximation algorithms for dispersion problems.
\newblock {\em Journal of algorithms}, 38(2):438--465, 2001.

\bibitem{cohen2016local}
V.~Cohen-Addad, P.~N. Klein, and C.~Mathieu.
\newblock Local search yields approximation schemes for $k$-means and
  $k$-median in {Euclidean} and minor-free metrics.
\newblock In {\em 57th Annual Symposium on Foundations of Computer Science
  (FOCS)}, pages 353--364. IEEE, 2016.

\bibitem{cunningham2015linear}
J.~P. Cunningham and Z.~Ghahramani.
\newblock Linear dimensionality reduction: survey, insights, and
  generalizations.
\newblock {\em Journal of Machine Learning Research}, 16(1):2859--2900, 2015.

\bibitem{dasgupta2008random}
S.~Dasgupta and Y.~Freund.
\newblock Random projection trees and low dimensional manifolds.
\newblock In {\em Proceedings of the 40th Symposium on Theory of Computing},
  pages 537--546. ACM, 2008.

\bibitem{fekete2004maximum}
S.~P. Fekete and H.~Meijer.
\newblock Maximum dispersion and geometric maximum weight cliques.
\newblock {\em Algorithmica}, 38(3):501--511, 2004.

\bibitem{fernandez2002polynomial}
W.~Fernandez de~la Vega, M.~Karpinski, and C.~Kenyon.
\newblock A polynomial time approximation scheme for metric {MIN-BISECTION}.
\newblock {\em Electronic Colloquium on Computational Complexity (ECCC)}, pages
  1--12, 2002.

\bibitem{friggstad2016local}
Z.~Friggstad, M.~Rezapour, and M.~R. Salavatipour.
\newblock Local search yields a {PTAS} for $k$-means in doubling metrics.
\newblock In {\em 57th Annual Symposium on Foundations of Computer Science
  (FOCS)}, pages 365--374. IEEE, 2016.

\bibitem{gollapudi2009axiomatic}
S.~Gollapudi and A.~Sharma.
\newblock An axiomatic approach for result diversification.
\newblock In {\em 18th International Conference on World Wide Web (WWW)}, pages
  381--390. ACM, 2009.

\bibitem{gottlieb2015nonlinear}
L.~A. Gottlieb and R.~Krauthgamer.
\newblock A nonlinear approach to dimension reduction.
\newblock {\em Discrete \& Computational Geometry}, 54(2):291--315, 2015.

\bibitem{har2011geometric}
S.~Har-Peled.
\newblock {\em Geometric approximation algorithms}, volume 173.
\newblock American mathematical society Boston, 2011.

\bibitem{hassin1997approximation}
R.~Hassin, S.~Rubinstein, and A.~Tamir.
\newblock Approximation algorithms for maximum dispersion.
\newblock {\em Operations Research Letters}, 21(3):133--137, 1997.

\bibitem{indyk2014composable}
P.~Indyk, S.~Mahabadi, M.~Mahdian, and V.~S. Mirrokni.
\newblock Composable core-sets for diversity and coverage maximization.
\newblock In {\em 33rd ACM Symposium on Principles of Database Systems}, pages
  100--108, 2014.

\bibitem{indyk2007nearest}
P.~Indyk and A.~Naor.
\newblock Nearest-neighbor-preserving embeddings.
\newblock {\em ACM Transactions on Algorithms (TALG)}, 3(3):31, 2007.

\bibitem{qin2012diversifying}
L.~Qin, J.~X. Yu, and L.~Chang.
\newblock Diversifying top-$k$ results.
\newblock {\em Proceedings of the VLDB Endowment}, 5(11):1124--1135, 2012.

\bibitem{radlinski2006improving}
F.~Radlinski and S.~Dumais.
\newblock Improving personalized web search using result diversification.
\newblock In {\em 29th SIGIR Conference on Research and Development in
  Information Retrieval}, pages 691--692. ACM, 2006.

\bibitem{ravi1994heuristic}
S.~S. Ravi, D.~J. Rosenkrantz, and G.~K. Tayi.
\newblock Heuristic and special case algorithms for dispersion problems.
\newblock {\em Operations Research}, 42(2):299--310, 1994.

\bibitem{tenenbaum2000global}
J.~B. Tenenbaum, V.~de~Silva, and J.~C Langford.
\newblock A global geometric framework for nonlinear dimensionality reduction.
\newblock {\em Science}, 290(5500):2319--2323, 2000.

\bibitem{vasconcelos2003feature}
N.~Vasconcelos.
\newblock Feature selection by maximum marginal diversity.
\newblock In {\em Advances in Neural Information Processing Systems}, pages
  1375--1382, 2003.

\bibitem{vieira2011query}
M.~R. Vieira, H.~L. Razente, M.~C.~N. Barioni, M.~Hadjieleftheriou,
  D.~Srivastava, C.~Traina, and V.~J. Tsotras.
\newblock On query result diversification.
\newblock In {\em 27th International Conference on Data Engineering (ICDE)},
  pages 1163--1174. IEEE, 2011.

\bibitem{wang1988study}
D.W. Wang and Y.S. Kuo.
\newblock A study on two geometric location problems.
\newblock {\em Information processing letters}, 28(6):281--286, 1988.

\end{thebibliography}

\end{document}